\documentclass[twoside]{article}
\usepackage{amsmath,amssymb,amsthm}
\usepackage{color}
\usepackage{graphicx} 
\usepackage{authblk}

\usepackage{hyperref}
\hypersetup{
	colorlinks=true,
	linkcolor=blue,
	citecolor=green
}
\newcommand{\nwc}{\newcommand}

\nwc{\xx}{\xi}

\nwc{\fix}[1]{\textcolor{red}{[#1]}}
\nwc{\note}[1]{\textcolor{blue}{#1}}

\nwc{\ip}[1]{\langle{#1}\rangle}
\nwc{\iph}[1]{\langle{#1}\rangle_{H^1}}

\nwc{\pow}{\alpha}
\nwc{\vp}{\varphi}
\nwc{\intR}{\int_{-\infty}^\infty}
\nwc{\fpl}{\frac{\pi}{L}}

\nwc{\R}{{\mathbb{R}}}
\nwc{\N}{{\mathbb{N}}}
\nwc{\Z}{{\mathbb{Z}}}
\nwc{\C}{{\mathbb{C}}}
\nwc{\D}{\partial}
\nwc{\eps}{\epsilon}
\nwc{\calA}{{\mathcal A}}
\nwc{\calB}{{\mathcal B}}
\nwc{\calF}{{\mathcal F}}
\nwc{\calG}{{\mathcal G}}
\nwc{\calH}{{\mathcal H}}
\nwc{\calI}{{\mathcal I}}
\nwc{\calJ}{{\mathcal J}}
\nwc{\calK}{{\mathcal K}}
\nwc{\calL}{{\mathcal L}}
\nwc{\calM}{{\mathcal M}}
\nwc{\calN}{{\mathcal N}}
\nwc{\calP}{{\mathcal P}}
\nwc{\calQ}{{\mathcal Q}}
\nwc{\calR}{{\mathcal R}}
\nwc{\calS}{{\mathcal S}}
\nwc{\calT}{{\mathcal T}}
\nwc{\calX}{{\mathcal X}}
\nwc{\calZ}{{\mathcal Z}}
\DeclareMathOperator{\sgn}{sgn}
\DeclareMathOperator{\sinc}{sinc}

\DeclareMathOperator{\re}{Re}

\nwc{\coloneq}{\colon=}
\nwc{\inv}{^{-1}}
\nwc{\dds}[1]{\frac{d#1}{ds}}
\nwc{\qua}{{\ \ }}

\nwc{\Heven}{H_{\rm even}}
\nwc{\Honeeven}{H_{\rm even}^1}
\nwc{\Ltwoeven}{L_{\rm even}^2}

\newtheorem{theorem}{Theorem}[section]

\newtheorem{corollary}[theorem]{Corollary}

\newtheorem{lemma}[theorem]{Lemma}
\newtheorem{proposition}[theorem]{Proposition}

\theoremstyle{remark}

\numberwithin{equation}{section}
\begin{document}

\title{Existence of solitary waves in particle lattices with power-law forces}
\author{Benjamin Ingimarson\footnote{Email address: ingimars@usc.edu} }
\author{Robert L. Pego\footnote{Email address: rpego@cmu.edu} }

\affil{%
Department of Mathematical Sciences\\ 
Carnegie Mellon University\\
Pittsburgh, PA 15213.}

\date{October 6, 2024}

\maketitle
\begin{abstract}
  We prove the existence of small solitary waves for one-dimensional lattices of particles
  that each repel every other particle with a force that decays as a power of distance.
 For force exponents $\alpha+1$ with $\frac43<\alpha<3$,
 we employ fixed-point arguments to find near-sonic solitary waves having scaled velocity profiles close to 
 non-degenerate solitary-wave profiles of fractional KdV or generalized Benjamin-Ono equations.
 These equations were recently found to approximately govern unidirectional long-wave motions in these lattices.
\end{abstract}


\section{Introduction}

In this work we prove an existence theorem for solitary waves of small amplitude
in an infinite lattice of particles which all interact with each other 
through long-range power-law forces.  
The equations of evolution that govern the particle positions $x_j$ 
as a function of time $t$ are
\begin{equation}\label{e:sys1}
    \ddot x_j = 
    -\pow \sum_{m=1}^\infty  
    \Bigl(  (x_{j+m}-x_j)^{-\pow-1} -(x_j-x_{j-m})^{-\pow-1}  \Bigr)\,,
    \quad j\in\Z.
\end{equation}
We require $x_j$ to strictly increase with $j$. 
E.g., $x_j = jh$  represents 
a rest state for a lattice of particles with uniform spacing $h>0$. 
The equations in \eqref{e:sys1} remain invariant under the scaling 
$x_j\mapsto h x_j$, $t\mapsto h^{(\pow+1)/2}t$ for $h>0$, as well as the Galilean group of translations and uniform motions.

When $1<\alpha<3$, we showed in a previous work \cite{IP24} that  
the unidirectional propagation of long-wave solutions of \eqref{e:sys1} 
is formally governed by  the nonlocal dispersive PDE 
\begin{equation}\label{e:BOalpha}
    \D_t u + u\D_x u + H|D|^\alpha u = 0,
\end{equation}
where $H$ is the Hilbert transform and the dispersion term
$f=H|D|^\alpha u$ has Fourier transform $\hat f(k)=(-i\sgn k)|k|^\alpha\hat u(k)$.
Subsequently, Wright~\cite{Wright.24} has rigorously proved that
long-wave solutions of \eqref{e:sys1} are close to solutions of
\eqref{e:BOalpha} over a suitable long time-scale provided $\alpha_*<\alpha<3$,
where $\alpha_*\approx 1.48$.

It is our aim in the present paper to prove that the system \eqref{e:sys1}
admits exact solitary wave solutions for speeds slightly exceeding 
the `sound speed' $c_\alpha$, 
which is the maximum speed of linear waves and is given by 
\begin{equation}\label{d:calpha}
 c_\alpha = \sqrt{\alpha(\alpha+1)\zeta_\alpha}\,,
\end{equation}
where $\zeta_s=\sum_{n=1}^\infty n^{-s}$ denotes the Riemann zeta function.
We will find such waves by approximation to solitary waves of \eqref{e:BOalpha}, 
which were first proved to exist by Weinstein~\cite{Weinstein.87} and 
Benjamin {\it et al.}~\cite{Benjamin.Bona.ea.90}.
The waves that we approximate need to have a non-degeneracy property
proved for a class of solutions including ground states
by Frank and Lenzmann \cite{Frank.Lenzmann.13}. 

In the case $\alpha=2$, the system \eqref{e:sys1} is an infinite
Calogero-Moser system. For this case, in \cite{IP24} we also established 
explicit formulas providing solitary waves 
having any supersonic speed $c>c_2=\pi$. 
The proof exploited some of the well-known completely integrable
structure of finite Calogero-Sutherland systems
to find periodic waves. 

Our present study builds instead on the formulation and methods
devised by Herrmann and Mikikits-Leitner in \cite{HML16}
in order to find solitary waves that approximate KdV solitons
for particle lattices with forces of any finite range.
The work \cite{HML16} in turn improved and simplified the method 
earlier employed by Friesecke and Pego in \cite{Friesecke.Pego.99} 
to obtain such a result for Fermi-Pasta-Ulam-Tsingou (FPUT) lattices,
which are lattices with nearest-neighbor forces.
Vainchtein~\cite{Vainchtein.22} has recently reviewed the literature 
concerning  solitary waves in particle lattices of various kinds.

We seek solitary waves having the following form:
\begin{equation}\label{d:xj}
x_j(t) = j- \eps^{\nu} U(\xx) , 
\quad \xx = \eps(j-c t), \quad c^2=c_\alpha^2+\eps^\mu.
\end{equation}
Consistent with the formal long-wave scaling found in \cite[Thm.~2.1]{IP24}
we take 
\[
\mu=\alpha-1, \qquad  \nu=\alpha-2. 
\]
For waves of the form in \eqref{d:xj}, the particle velocity $\dot x_j(t)=c\eps^\mu U'(\xx)$.
As we discuss in Section~\ref{s:equations} below, 
the scaled velocity profile $W=U'$ needs to satisfy a nonlocal, nonlinear eigenvalue problem, 
which formally reduces in the limit $\eps\to0$ 
to a nonlocal quadratic equation, namely
\begin{equation}\label{e:W0eq-intro}
    W + \kappa_3 |D|^\mu W = \tfrac12\kappa_2 W^2 \,,
\end{equation}
where $\kappa_3$ and $\kappa_2$ are positive constants as found in \cite{IP24};
see \eqref{e:betaalpha2} and \eqref{d:Q0} below.
Solutions of \eqref{e:W0eq-intro} provide solitary waves of \eqref{e:BOalpha}
after appropriate scaling.
A profile $W$ satisfying \eqref{e:W0eq-intro} is called {\em non-degenerate} if the linearized operator
\begin{equation}
    L_+ = I + \kappa_3|D|^\mu - \kappa_2 W \,,
\end{equation}
acting in $L^2(\R)$, has one-dimensional kernel spanned by the derivative $W'$.

As stated by Frank and Lenzmann \cite[p.~262]{Frank.Lenzmann.13},
for \eqref{e:W0eq-intro} to admit any solution having finite energy 
(i.e., in $H^{\mu/2}(\R)\cap L^3(\R)$), it is necessary that
\begin{equation}\label{c:alpha}
    \tfrac43<\alpha<3 \,,
\end{equation}
due to Pohozaev identities. 
(See \cite[sec.~3.5]{Ambrosio.book} for the key to the nontrivial proof of these identities.)
Consequently our results for $1<\alpha<3$ will be restricted to the smaller range in \eqref{c:alpha}.
For all $\alpha$ in this smaller range, however, ground-state solutions 
(positive, even, energy-minimizers) 
exist and are proved in \cite{Frank.Lenzmann.13} to be non-degenerate.
Moreover, any solution of \eqref{e:W0eq-intro} must be positive, as discussed in Section~\ref{s:positive}
below.

To find profiles of solitary waves of \eqref{e:sys1},
similar to \cite{Friesecke.Pego.99} and \cite{HML16} 
we formulate a fixed-point equation and regard it as a perturbation of
a corresponding fixed-point equation 
for solutions of \eqref{e:W0eq-intro}. 
We analyze the fixed-point equations, however, 
in the space of even functions in $H^1(\R)$, rather than in $L^2(\R)$ as was done in \cite{HML16}.
This has the natural advantage of working in a Banach algebra of functions,
and we obtain further simplification by initially seeking less precise control
over the size of the correction.
In principle, spectral analysis of the linearization of the fixed-point equation 
could have become more complicated in $H^1(\R)$ instead of $L^2(\R)$. 
But we were able to substantially simplify spectral analysis in $H^1(\R)$ 
by extracting  from \cite{HML16}  a key compactness argument 
and casting it into an abstract form; see Lemma~\ref{lem:strongnorm} below.

The plan of this paper is as follows. In Section~\ref{s:prelim} we develop preliminaries.
We derive the fixed-point equations governing solitary wave profiles for \eqref{e:sys1} and its formal limit \eqref{e:W0eq-intro}
in Section~\ref{s:equations} and precisely state the main theorem. 
In Section~\ref{s:fixedpt} we prove the existence of solitary wave profiles for \eqref{e:sys1}. 
For $\frac43<\alpha<3$, given any non-degenerate even solution $W_0\in H^1(\R)$ of \eqref{e:W0eq-intro}, 
when $c-c_\alpha$ is positive and small we find a (locally unique) scaled profile $W_\eps=U'$ that is even, positive, and close to $W_0$ in $H^1$,
providing a solitary wave for \eqref{e:sys1} as in \eqref{d:xj}.
We carry out a fixed-point analysis based on a quantitative fixed-point lemma from \cite{Friesecke.Pego.99}.
Control over the deviation $\|W_\eps-W_0\|_{H^1}$ comes by adapting the rigorous residual estimates of Wright~\cite{Wright.24}.
When $\alpha=2$ and $\eps$ is sufficiently small, the waves we find here agree with the ones 
provided by the implicit formulas in \cite{IP24}; see Subsection~\ref{ss:CM}.

We establish positivity and smoothness of the profiles $W_\eps$ in Section~\ref{s:positive}. There we also show that
the unscaled velocity profiles, given by $v_c(z)=c\eps^\mu W_\eps(\eps z)$, are analytic as functions of wave speed.
We study a Hamiltonian energy $\calH$ for the solitary waves of~\eqref{e:sys1} in Section~\ref{s:hamiltonian}. 
For $\alpha=2$ we find explicit formulas by using results from \cite{IP24}.
The sign of $d\calH/dc$ agrees with the sign of $\alpha-\frac32$ when the latter is non-zero, 
for sufficiently small $\eps$ depending on $\alpha$. 
In a variety of lattice wave and other Hamiltonian wave stability problems, 
a change in the sign of $d\calH/dc$ has been associated with transitions to  instability; 
e.g.,  see~\cite{GSSI.1987,GSSII.1990,Friesecke.Pego.2002,Friesecke.Pego.2004a,Vainchtein.22,Cuevas-MaraverEA.2017,DuranEA2022}.
Whether this may be the case for systems such as \eqref{e:sys1} remains an open problem. 

The value $\alpha=\frac32$ is $L^2$-critical for \eqref{e:W0eq-intro}.
In this regard it is curious that Wright's result in \cite{Wright.24}, showing that solutions 
of \eqref{e:BOalpha} approximate long-wave solutions of \eqref{e:sys1} over long times,
is valid for all $\alpha$ in a neighborhood of $\frac32$.

In the interest of brevity, we do not address the range $\alpha\ge3$ in the present paper.  In that range naturally one expects a KdV limit, 
but also one should be able to treat a much more general family of interparticle forces. 
In particular, the case of alternating signs studied formally in \cite{IP24} seems particularly challenging and deserves a separate study.

\section{Preliminaries}\label{s:prelim}
The Maclaurin series for 
$Z(r):=\alpha(1-r)^{-\alpha-1}$ takes the form
\begin{equation} \label{e:Zseries}
    Z(r) =  \sum_{k=0}^\infty \alpha_k r^k\,,
    \qquad
    \alpha_k = 
     \frac{\alpha(\alpha+1)\cdots(\alpha+k)}{k!} \,.
\end{equation}
The coefficients $\alpha_k$ are defined differently than in \cite{IP24}
for present convenience.

In the standard Sobolev space $H^s=H^s(\R)$, $s\ge0$, we use the inner product
given in terms of the Fourier transform 
$\hat f(k)=\frac1{2\pi}\int_\R f(\xx)e^{-ik\xx}\,d\xx$ 
by 
\begin{equation}\label{e:Hs_ip}
 \ip{f,g}_{H^s} = \int_{\R} (1+|k|^2)^s \hat f(k)\overline{\hat g(k)}\,dk\,.
\end{equation}
We recall that  for each $s>\frac12$,  there is a constant $C_{H^s}\ge1$ such that 
\begin{equation}\label{e:Balg}
   \|fg\|_{H^s}  \le C_{H^s} \|f\|_{H^s}\|g\|_{H^s} \quad\text{for all $f,g\in H^s$.}
\end{equation}
We take the inner product in $L^2=L^2(\R)$ identical to that for $s=0$  above.

We let $\Ltwoeven$ denote the subspace of even elements of $L^2$,
elements $f$ for which $f(-\xx)=f(\xx)$ for a.e.~$\xx\in\R$ 
(or equivalently $\hat f(-k)=\overline{\hat f(k)}$ for a.e.~$k\in\R$),
and we let $\Heven^s=H^s\cap \Ltwoeven$.
The space $\calL(H^s)$ is the space of bounded linear operators on $H^s$,
equipped with the operator norm.

Following \cite{HML16}, we will make heavy use of the 
symmetric averaging operators $\calA_\eta$ defined for $\eta>0$ on $H^s$
for $s\ge0$ by 
\begin{equation}
    \calA_\eta f(\xx) = \frac1\eta\int_{-\eta/2}^{\eta/2} f(\xx+z)\,dz,
\end{equation}
which satisfy
\begin{equation}\label{e:Aetadiff}
   f(\xx+\tfrac12\eta) - f(\xx-\tfrac12\eta)  = 
   \eta \calA_\eta (\D_\xx f)(\xx)  =  \eta \D_\xx (\calA_\eta f)(\xx)  \, ,
\end{equation}
\begin{equation}\label{e:AetaFourier}
    \widehat{\calA_\eta f}(k) = \sinc(\tfrac12 \eta k) \hat f(k)\, .
\end{equation}
Here $\sinc(z)=(\sin z)/z$.  
From the Fourier representation it is clear that
the operators $\calA_\eta$ map $H^s$ into $H^{s+1}$ continuously.
Moreover since $\sinc(\frac12\eta k)$ lies in $[-1,1]$ and converges to $1$
as $\eta\to 0$ for any $k$, it is clear that 
$\calA_\eta$ is nonexpansive on $H^s$ and 
converges to the identity {\em strongly} (but not in operator norm).
That is, for any $f\in H^s$ we have
\begin{align}
& \|\calA_\eta f\|_{H^s} \le \|f\|_{H^s} \,,
\label{e:Aeta_bound}
\end{align}
and
\begin{equation}
\|\calA_\eta f-f\|^2_{H^s}  = 
    \int_\mathbb{R} \left(1+ k^{2}\right)^s \left| \sinc (\tfrac12\eta k) - 1\right|^{2} |\hat{f}(k)|^2 \, \mathrm dk
\to 0 
\label{e:Aeta_slim}
\end{equation}
as $\eta\to0$. Because $|1-\sinc z |\le \frac16 z^2$ for all $z$ it also follows 
\begin{equation}\label{e:Aeta_est}
\|\calA_\eta f-f\|_{H^s}  \le \frac{\eta^2}{24} \|f\|_{H^{s+2}}  
\end{equation}
for all $f\in H^{s+2}$. Note further that $\calA_\eta f$ is even if and only if $f$ is even.
Moreover, if $f$ is even and unimodal (i.e., even, and decreasing on $(0,\infty)$) then $\calA_\eta f$ is also,
since for $f$ smooth we have $\D_\xx(\calA_\eta f)\le0$ for $\xx\ge0$ by \eqref{e:Aetadiff}.

\section{Equations for solitary-wave profiles}\label{s:equations}
In this section, we follow the approach of Herrmann and Mikikits-Leitner in \cite{HML16} to formulate a fixed point equation 
whose solution provides velocity profiles of solitary waves for \eqref{e:sys1}.

\subsection{Equations for profiles on lattices}
Due to the ansatz~\eqref{d:xj}, by \eqref{e:Aetadiff} and since $\mu=\alpha-1=\nu+1$ we can write 
\begin{align*}
 x_{j+m}-x_j = m - \eps^\nu(U(\xx+\eps m)-U(\xx)) 
  = m(1-\eps^\mu \calA_{m\eps}W(\xx+\tfrac12m\eps)),\\
 x_{j}-x_{j-m} = m - \eps^\nu(U(\xx)-U(\xx-\eps m)) 
  = m(1-\eps^\mu \calA_{m\eps}W(\xx-\tfrac12m\eps)),
\end{align*}
where 
\[
W = U'.
\]
By consequence,
\begin{align*}
    \alpha(x_{j+m}-x_j)^{-\alpha-1} &= Z(\eps^\mu\calA_{m\eps}W(\xx+\tfrac12m\eps)) \,m^{-\alpha-1},\\
    \alpha(x_j-x_{j-m})^{-\alpha-1} &=  Z(\eps^\mu\calA_{m\eps}W(\xx-\tfrac12m\eps))\,m^{-\alpha-1},
\end{align*}
and after taking the difference and using formula \eqref{e:Aetadiff} again,
we find that for system \eqref{e:sys1} to be satisfied 
it is necessary and sufficient that
\begin{align} \label{e:nevp0}
c^2\eps^{\alpha} \,\D_\xx W(\xx) 
&= \sum_{m=1}^\infty \frac{m\eps}{m^{\alpha+1}} 
  \D_\xx\calA_{m\eps} Z(\eps^\mu\calA_{m\eps}W)
\,.
\end{align}
We seek (weak) solutions of this equation in $H^1$. By requiring that 
\begin{equation}\label{c:eWlt1}
\eps^\mu C_{H^1}\|W\|_{H^1}<1 \,, 
\end{equation}
we ensure that the MacLaurin series for $Z(\eps^\mu\calA_{m\eps}W)-Z(0)$ converges in $H^1$,
avoiding the singularity of $Z(r)$ at $r=1$.
Since $Z(0)=\alpha$, 
we thus find it necessary and sufficient that $W$ should satisfy 
the nonlocal nonlinear eigenvalue problem 
\begin{equation}\label{e:nevp}
c^2 W = \sum_{m=1}^\infty \frac{\eps^{-\mu}}{m^\alpha} 
\calA_{m\eps} (Z(\eps^\mu\calA_{m\eps}W)-\alpha) .
\end{equation}
As in \cite{HML16}, we recast this equation by collecting linear terms
on the left-hand side and separating the quadratic terms. 
Recall that $c^2=\eps^\mu+\alpha_1\sum_{m\ge1}m^{-\alpha}$ 
from \eqref{d:xj} and \eqref{d:calpha},
and define
\begin{align}
Z_3(r) &= \alpha(1-r)^{-\alpha-1}-\alpha-\alpha_1 r -\alpha_2r^2,
\label{d:Z3}
\\
\calB_\eps W &= W + \alpha_1\sum_{m=1}^\infty \frac{\eps^{-\mu}}{m^\alpha}(W - \calA_{m\eps}^2 W),
  \label{d:Beps}
\\
\calQ_\eps(W) &= \alpha_2\sum_{m=1}^\infty \frac{1}{m^\alpha}\calA_{m\eps}(\calA_{m\eps}W)^2,
  \label{d:Qeps}
\\
\calZ_\eps(W) &= 
  \sum_{m=1}^\infty \frac{\eps^{-2\mu}}{m^\alpha}\calA_{m\eps} Z_3(\eps^\mu\calA_{m\eps}W).
  \label{d:Zeps}
\end{align}
After substitution and further dividing by $\eps^\mu$ we find \eqref{e:nevp} equivalent to 
\begin{equation}\label{e:evp3}
{\calB_\eps W = \calQ_\eps(W) + \calZ_\eps(W).}
\end{equation}

\subsection{Formal limit equations}

The operator $\calB_\eps$ is a Fourier multiplier,
with $\widehat{\calB_\eps W}(k) = b_\eps(k)\hat W(k)$ 
where the symbol $b_\eps$ is given by 
\begin{equation} \label{d:beps}
    b_\eps(k) = 1+ \alpha_1 \eps \sum_{m=1}^\infty  
    \frac{1-\sinc^2(\frac12 km\eps)}{(m\eps)^\alpha} \,.
\end{equation}
We have 
\begin{equation}\label{e:beps_bound1}
1\le b_\eps(k)\le \eps^{-\mu}\alpha_1\zeta_\alpha \quad\text{ for all $k$},
\end{equation} 
so $\calB_\eps$ is bounded on $H^s$ for any fixed $\eps>0$, with nonexpansive inverse $\calB_\eps\inv$.

For $1<\alpha<3$, similar to what was noted in \cite{IP24},
the sum in \eqref{d:beps} approximates a convergent integral. Indeed, as $h\to0^+$,
\begin{equation}\label{e:Slim}
S_\alpha(h):= \alpha_1 h \sum_{m=1}^\infty \frac{1-\sinc^2(\frac12 mh)}{(mh)^\alpha}
\to \kappa_3 \,,
\end{equation}
where, with notation consistent with \cite[Thm.~2.1]{IP24},
\begin{equation}\label{d:kappa3}
\kappa_3
:= \alpha_1\int_0^\infty \frac{1-\sinc^2(z/2)}{z^\alpha}\,dz\,.
\end{equation}
Since $b_\eps(k)=b_\eps(|k|)=1+ |k|^{\alpha-1}S_\alpha(\eps |k|)$,  we have that for each fixed $k\in\R$,
\begin{equation}\label{e:bepslim}
    b_\eps(k) \to b_0(k) :=  1+\kappa_3 |k|^{\alpha-1}  \quad\text{as $\eps\to0$}.
    \end{equation}
We let $\calB_0$ denote the Fourier multiplier with symbol $b_0(k)$, writing
\begin{equation}\label{d:B0}
    \calB_0 W = W + \kappa_3|D|^{\alpha-1}W\,.
\end{equation}
We remark that due to the formula for the integral in \cite[Remark 3]{IP24}, we have
\begin{equation}\label{e:betaalpha2}
\kappa_3=
\begin{cases}
    \pi, & \alpha=2,\\
    -2\sin(\frac12\pi\alpha)\Gamma(1-\alpha), & \alpha\in(1,2)\cup(2,3).
\end{cases}
\end{equation}

For the quadratic term in \eqref{e:evp3}, we find 
$\|\calA_{m\eps}(\calA_{m\eps}f)^2 - f^2\|_{H^s}\to0$ 
as $\eps\to0$ for any $f\in H^s$ with $s>\frac12$, 
by using the strong convergence property \eqref{e:Aeta_slim}
and the fact that $H^s$ is a Banach algebra.
Defining 
\begin{equation}\label{d:Q0}
 \calQ_0(f) := \tfrac12 \kappa_2 f^2, \qquad \kappa_2 = 2\alpha_2\zeta_\alpha\,,
\end{equation}
it follows that as $\eps\to0$,
\begin{equation}\label{e:Qeps_slim}
\calQ_\eps(f) - \calQ_0(f) = \alpha_2\sum_{m=1}^\infty m^{-\alpha}
\left(\calA_{m\eps}(\calA_{m\eps}f)^2 - f^2\right)
\to 0 
\end{equation}
in $H^s$ norm,
by applying the dominated convergence theorem to the sum after taking norms term by term.

We will establish a rigorous bound on the higher-order term $\calZ_\eps(W)$ later. 
For now, we note that it is formally $O(\eps^\mu)$ since $Z_3(r)=O(r^3)$.
Thus we expect that as $\eps\to0$, \eqref{e:evp3} should approximate equation \eqref{e:W0eq-intro},
which we can recast in the form
\begin{equation}\label{e:W0eq}
{\calB_0 W = \calQ_0(W)\,.}
\end{equation}
This equation determines the profile of solitary waves of speed $\tilde c=1/\kappa_1$ 
of the nonlocal dispersive equation
\begin{equation}
    \kappa_1\D_t u + \kappa_2 u\D_x u +\kappa_3 H|D|^\alpha u = 0\,.
\end{equation}
According to \cite[Thm.~2.1]{IP24} and the rigorous results of Wright~\cite{Wright.24}, 
this equation, with $\kappa_1=2c_\alpha$, 
is the correctly scaled formal limit of \eqref{e:sys1} 
consistent with the long-wave ansatz 
\[
 x_j=j+\eps^\nu v(\eps(j-c_\alpha t),\eps^\alpha t)\,,  \quad u = -\D_\xx v.
 \]
\subsection{Fixed-point formulation and main result}
Similar to \cite{HML16}, our approach to find solutions of \eqref{e:evp3}
is to fix a known {\em even} solution $W_0$ to \eqref{e:W0eq}, meaning
an even solution of the fixed-point equation
\begin{equation}\label{d:G0}
    W_0 = \calF(W_0):= \calB_0\inv \calQ_0(W_0)\,,
\end{equation}
and  solve the fixed-point corresponding to \eqref{e:evp3},
which is 
\begin{equation}\label{d:Geps}
    W = \calG_\eps(W):= \calB_\eps\inv (\calQ_\eps(W)+\calZ_\eps(W))\,,
\end{equation}
through a perturbation analysis. 
We will suppose that $W_0\in\Heven^1$ is a 
given solution of \eqref{e:W0eq} that is non-degenerate. 
Recall this means that  the linearized operator 
\begin{equation}
\calL_+=\calB_0-D\calQ_0(W_0)  \,,
\end{equation}
acting in $L^2$, has one-dimensional kernel spanned by the odd function $W_0'$.
By a bootstrapping argument, it follows that $W_0 = \calB_0\inv (\frac12\kappa_2 W_0^2)$
belongs to $H^s_{\rm even}$ for all $s>0$, hence is smooth. 
Moreover, $W_0$ is positive, since $W_0^2$ is positive and the Green's function for the operator 
$\calB_0$ is positive (see Section~\ref{s:positive} below).

Our main results are stated precisely as follows. 

\begin{theorem}\label{t.main} 
Assume $\frac43<\alpha<3$ and $W_0\in H^1$ is an even solution of \eqref{e:W0eq}
that is non-degenerate. Then there exist positive constants $\eps_0$, $\delta$,
and $C$ such that the following hold, for each $\eps\in(0,\eps_0)$: 
\begin{itemize}
    \item[(i)]  $\calG_\eps$
    has a unique fixed point $W_\eps\in \Heven^1$ satisfying $\|W_\eps-W_0\|_{H^1}\le\delta$.
\item[(ii)] $\|W_\eps-W_0\|_{H^1}\le C\eps^\gamma, 
\quad\text{where}\quad
\gamma = \begin{cases}
\alpha-1 &\alpha\in(1,2],\\ 
3-\alpha&\alpha\in (2,3).\end{cases}$
    \item[(iii)] $W_\eps$ is everywhere positive.
    \item[(iv)] $W$ is smooth, with $W\in H^\infty$.
\end{itemize}
Furthermore, the map  $c\mapsto v_c\in\Heven^1$,
from wave speed $c$ to the 
unscaled velocity profile $v_c$ given by 
\begin{equation}\label{d:vc}
   v_c(z) := c\eps^\mu W_\eps(\eps z), \qquad c^2 = c_\alpha^2+\eps^\mu,
\end{equation}
is analytic. 
\end{theorem}

Note that the unscaled velocity profile function $v_c$ determines the 
particle velocities according to 
$\dot x_j(t) = v_c(j-ct)$, cf.~\eqref{d:xj}.
Note as well that given $f\in \Heven^1$, 
the dilation map $\eps\mapsto f(\eps\cdot)$ 
may not be analytic, or even differentiable;
thus we do not discuss the regularity of the map $\eps\mapsto W_\eps$.

\section{Fixed-point analysis}\label{s:fixedpt}

In order to prove the existence of wave profiles as fixed points in equation~\eqref{d:Geps},
we make use of the quantitative version of the standard inverse function theorem 
stated as Lemma~A.1 in \cite{Friesecke.Pego.99} and proved there.
Restated for clarity, it takes the following form, in which $\|\cdot\|$ denotes
the norm in $E$ or the operator norm on $\calL(E)$ as appropriate.
\begin{lemma}\label{lem:IFT}
    Let $F$ and $G$ be $C^1$ maps from a ball $B$ in a Banach space $E$ to $E$.
    Suppose $u_0=F(u_0)$ and that $L=I-DF(u_0)$ is invertible with operator norm
    $\|L\inv\|\le C_0<\infty$.
   Assume that positive constants $C_1$, $C_2$, $\theta$ and $\delta$ satisfy
   \begin{equation}\label{c:012}
   C_0(C_1+ C_2)\le \theta<1  \,,
   \end{equation}
\begin{equation}
      \|F(u_0)-G(u_0)\|\le \delta(1-\theta)/C_0 \,,
      \label{c:DF1}
\end{equation}
and that whenever $\|u-u_0\|\le\delta$, $u$ is in the ball $B$ and
  \begin{align}
      \|DF(u)-DF(u_0)\| &\le C_1\,, 
      \label{c:DF2}
      \\   \|DF(u)-DG(u)\| &\le C_2\,.
      \label{c:DF3}
  \end{align} 
  Then $u=G(u)$ for some unique $u\in B$ satisfying $\|u-u_0\|\le\delta$, and moreover
  \[
  \|u-u_0\| \le C_0(1-\theta)\inv  \|F(u_0)-G(u_0)\| \,.
  \]
\end{lemma}
We will apply this lemma to the functions $F=\calF$ and $G=\calG_\eps$ on
$E=\Heven^1$, for $\eps>0$ sufficiently small.
The functions $\calF$ and $\calG_\eps$ will be shown to be analytic on a suitable ball in $H^1$.
The operator 
\begin{equation}
    \calL_0 = I - D\calF(W_0)  = I - \calB_0\inv D\calQ_0(W_0)
\end{equation}
will be shown to be Fredholm, and is invertible because $W_0$ is non-degenerate.
Establishing \eqref{c:DF2} will be easy. To obtain the residual estimate \eqref{c:DF1} we will use 
rigorous residual bounds established by Wright \cite{Wright.24}.
Our proof of \eqref{c:DF3} involves a contradiction argument based on 
a key compactness property. This is essentially a distillation of 
Herrmann \& Mikikits-Leitner's  proof in \cite{HML16} 
of invertibility in $\Ltwoeven$ for an  operator analogous to $I-D\calG_\eps(W_0)$,
uniformly for all small enough $\eps>0$.
In the present context, uniform invertibility in $\Heven^1$ follows from 
conditions \eqref{c:012}--\eqref{c:DF3} together with a Neumann series expansion.

\subsection{Analyticity and symmetry}
We first establish the analyticity of various maps on $H^1$,
referring to \cite[Ch.~2.3]{Berger} for the basic theory of analytic maps on Banach spaces.
Note the maps $\calQ_0$ and $\calQ_\eps$ are continuous quadratic maps on $H^1$, hence are analytic.
Any monomial map $f \mapsto f^k$ is analytic, and compositions and uniform limits of analytic functions 
are analytic.  Regarding $\calZ_\eps$ we have the following. 

\begin{lemma}\label{lem:Zeps}
   Let $\eps,\rho\in(0,1)$, and for $R>0$ let 
   \[
   \tilde B_R = \{f\in H^1: C_{H^1}\|f\|_{H^1}\le R\}
   \]
   be the closed ball of radius $R/C_{H^1}$ in $H^1$. Then $\calZ_\eps\colon\tilde B_R\to H^1$ is analytic
   provided $\eps^\mu R\le\rho<1$, and the following bounds hold for all $f\in \tilde B_R$:
   \[
   \|\calZ_\eps(f)\|_{H^1} \le \eps^\mu\zeta_\alpha Z_3(\rho)
   \left(\frac{R}\rho\right)^{3}, \qquad 
   \|D\calZ_\eps(f)\|_{\calL(H^1)}  \le  \eps^\mu \zeta_\alpha Z_3'(\rho)
   \left(\frac{R}\rho\right)^{2}.
   \]
\end{lemma}
\begin{proof}
Recall the series expansion for $Z_3(r)=\sum_{k=3}^\infty \alpha_k r^k$ converges for $|r|<1$.
Since $\|f^k\|_{H^1}\le (C_{H^1}\|f\|_{H^1})^k$ for all $k$, it follows that the 
Nemytskii operator $f\mapsto Z_3\circ f$ is analytic on the ball $\tilde B_\rho$ provided $\rho<1$,
with  $\|Z_3\circ f\|_{H^1} \le Z_3(\rho)$ for all $f\in \tilde B_\rho$.

Thus, for any $R>0$ and each $m\ge0$, the map $W\mapsto\calA_{m\eps}Z_3(\eps^\mu\calA_{m\eps}W)$ 
is analytic on  the ball $\tilde B_R$ provided $\eps^\mu R\le \rho$, and 
\[
\| \calA_{m\eps}Z_3(\eps^\mu\calA_{m\eps}W) \|_{H^1}  \le Z_3(\eps^\mu R) \le Z_3(\rho)\left(\frac{\eps^{\mu}R}{\rho}\right)^3  .
\]
It follows that the series expansion for $\calZ_\eps(W)$ in \eqref{d:Zeps} then converges uniformly in $H^1$ 
on $\tilde B_R$ under the same condition, with the stated bound on the $H^1$ norm.  

It is then straightforward to show in a similar way that for all $W\in\tilde B_R$ and $V\in H^1$,
since $Z_3'(\eps^\mu R)\le Z_3'(\rho)(\eps^\mu R/\rho)^2$,
\begin{equation}
    D\calZ_\eps(W)V = \eps^{-2\mu}\sum_{m\ge1} m^{-\alpha}\calA_{m\eps} (Z_3'(\eps^\mu\calA_{m\eps}W)(\eps^\mu \calA_{m\eps} V))\,,
\end{equation}
and 
\[
\|D\calZ_\eps(W)V\|_{H^1} \le 
\eps^\mu \zeta_\alpha Z_3'(\rho)\left(\frac R\rho\right)^2\|V\|_{H^1}\,.
\]
This finishes the proof.
\end{proof}

Regarding symmetry, we note that since the symbols of the operators $\calA_\eta$, $\calB_\eps$
and $\calB_0$ are real, even, and bounded, these operators map even functions in $H^s$ to even functions in $H^s$. 
For $s>\frac12$ the monomial maps $f\mapsto f^k$ also have the same property. From this and the lemma above we infer
the following.
\begin{proposition}\label{p:Fanalytic}
    The map $\calF$ in \eqref{d:G0} is an analytic map from $\Heven^1$ into itself.   
    For any $\eps\in(0,1)$ and $R>0$ such that $\eps^\mu R<1$, the map $\calG_\eps$ in \eqref{d:Geps}
    is analytic from $\tilde B_R\cap\Heven^1$ into $\Heven^1$.
\end{proposition}

\subsection{Residual estimates}
According to the definitions in \eqref{d:G0} and \eqref{d:Geps} we can write 
\begin{equation}
    \calF(W_0)-\calG_\eps(W_0) = \calR_\eps - \calS_\eps\,,
\end{equation}
where
\begin{align}
\calR_\eps &= 
\calB_0\inv\calQ_0(W_0)  -  \calB_\eps\inv\calQ_\eps(W_0)  \,, 
\quad \calS_\eps = \calB_\eps\inv\calZ_\eps(W_0) \,. 
\label{d:calRS}
\end{align}
We have $\|\calS_\eps\|_{H^1}\le \|\calZ_\eps(W_0)\|_{H^1}$
since $\calB_\eps\inv$ is non-expansive on $H^1$, 
so we find the following
by simply applying Lemma~\ref{lem:Zeps} and recalling $\mu=\alpha-1$.
\begin{lemma} For all $\eps>0$ sufficiently small we have 
$\|\calS_\eps\|_{H^1}\le C\eps^{\alpha-1}$.
\end{lemma}

For the term $\calR_\eps$ we claim the following.
\begin{proposition}\label{p:Reps} For all $\eps>0$ sufficiently small we have 
  \[
  \|\calR_\eps\|_{H^1}  \le 
   \begin{cases} 
     C\eps & \alpha\in (1,2],\\
     C \eps^{3-\alpha} & \alpha\in(2,3).
   \end{cases}
  \] 
\end{proposition}
The proof will be provided presently. 
But the last two results together immediately imply the following residual estimate.
\begin{corollary}\label{cor:res}
    For all sufficiently small $\eps>0$ we have 
    \[ \| \calF(W_0)-\calG_\eps(W_0) \|_{H^1} \le 
   \begin{cases} 
     C\eps^{\alpha-1} & \alpha\in (1,2],\\
     C \eps^{3-\alpha} & \alpha\in(2,3).
   \end{cases}
   \]
\end{corollary}

To prove Proposition~\ref{p:Reps} we adapt Wright's method of estimating residuals
in the long-wave approximation of \eqref{e:sys1} in \cite{Wright.24}.
We begin with an estimate on the difference of quadratic functions.
\begin{lemma}\label{lem:quad_diff}
    There exists $C>0$ independent of $\eps$ such that 
    \[
    \|\calQ_\eps(W_0)-\calQ_0(W_0)\|_{H^1} \le C\eps^{\alpha-1} \,.
    \]
\end{lemma}
\begin{proof}
First observe that 
\begin{equation}
  \| \calQ_\eps (W_0) - \calQ_0(W_0) \|_{H^{1}} \leq \sum_{m\geq 1} \frac{\alpha_2}{m^{\alpha}} 
  \| \calA_{m\eps} (\calA_{m\eps} W_0 )^{2}    - W_0^{2}\|_{H^{1}}
  \,. \label{e:dub_diff}
\end{equation}
We claim that for some constant $C$ independent of $\eps$ and $m$,
\begin{equation}\label{e:smear}
  \| \calA_{m\eps} (\calA_{m\eps} W_0 )^{2}    - W_0^{2}\|_{H^{1}}
  \le C m^2\eps^2.
\end{equation}
Indeed,  by the triangle inequality,
\begin{align*}
\| \calA_{m\eps} (\calA_{m\eps} W_0 )^{2}    - W_0^{2}\|_{H^{1}}
 \le \ &\| (\calA_{m\eps}-I) (\calA_{m\eps} W_0 )^{2}  \|_{H^1}
\\ & + \ \| (\calA_{m\eps} W_0 )^{2}   - W_0^{2}\|_{H^{1}} \,,
\end{align*}
from which one can infer \eqref{e:smear} by using \eqref{e:Aeta_est} and the $H^3$ regularity of $W_0$.

Using the estimate \eqref{e:smear} in \eqref{e:dub_diff} for small $m$, we find
\begin{align*}
  \sum_{m=1}^{\lfloor 1/\eps\rfloor} m^{-\alpha} \|\calA_{m\eps} (\calA_{m\eps} W_0)^{2} - W_0^{2}\|_{H^{1}} &\leq 
  \sum_{m=1}^{\lfloor 1/\eps\rfloor} m^{-\alpha} C(m\eps)^2
  \leq \frac{C\eps^{\alpha-1}}{3-\alpha} \,. 
\end{align*}
In the last line, we used a simple integral bound $\sum_{m=1}^{\lfloor 1/\eps\rfloor} m^{2-\alpha} \leq \frac{1}{3-\alpha}\eps^{\alpha -3}$ as in \cite{Wright.24}.
For large $m$, we simply bound the norms in \eqref{e:dub_diff} by a constant, and get through a similar integral bound 
\begin{equation*}
  \sum_{m > \lfloor 1/\eps \rfloor} \frac{\alpha_2}{m^{\alpha}} \|\calA_{m\eps} (\calA_{m\eps} W_0)^{2} - W_0^{2}\|_{H^{1}} \leq \frac{C'\eps^{\alpha-1}}{\alpha -1}\,, 
\end{equation*}
where $C'$ is another constant independent of $\eps$ and $m$. 
\end{proof}

Next we deal with estimates on differences of symbols and operators.
\begin{lemma}\label{lem:bepsdiff}
For all $\eps>0$ sufficiently small we have 
(i) For all $k\in\R$,
\[
   |b_\eps(k)\inv - b_0(k)\inv| \le
   \begin{cases} 
     C|k| \eps & \alpha\in (1,2],\\
     C|k|^{3-\alpha} \eps^{3-\alpha} & \alpha\in(2,3).
   \end{cases}
\]
(ii) For each $s\ge0$ and all $f\in H^{s+3}$, 
\[
\|(\calB_\eps\inv-\calB_0\inv)f\|_{H^s} \le 
   \begin{cases} 
     C\eps\, \|f\|_{H^{s+1}} & \alpha\in (1,2],\\
     C \eps^{3-\alpha} \|f\|_{H^{s+3-\alpha}}& \alpha\in(2,3).
   \end{cases}
\]
\end{lemma}
\begin{proof}
First we note the estimates
\begin{align}
S_\alpha(h) = h^{1-\alpha}& \sum_{m\geq 1} \frac{1-\sinc^{2}(mh/2)}{m^{\alpha}} 
\leq h^{1-\alpha}\zeta_\alpha \,, 
\label{e:Salphabd}
\\
 \left| b_\eps^{-1}(k) - b_0^{-1}(k)\right| &
 = \frac{|b_0(k)-b_\eps(k)|}
 {b_0(k)b_\eps(k)}
    = \frac{|k|^{\alpha-1} \left|S_\alpha(\eps k) - \kappa_3 \right|}
 {b_0(k)b_\eps(k)}
    \nonumber\\
    &\leq \frac{\left|S_\alpha(\eps k) - \kappa_3\right|}{\kappa_3 b_\eps(k)}\,,
\label{e:bdiffbd}
\end{align}
since by its definition in \eqref{e:bepslim}, $b_0(k)=1+\kappa_3 |k|^{\alpha-1}$.

Now we invoke Lemma 3 of \cite{Wright.24}, which directly implies that for all $h>0$,
\[
|S_\alpha(h)-\kappa_3| \le 
   \begin{cases} 
     Ch & \alpha\in (1,2],\\
     Ch^{3-\alpha} & \alpha\in(2,3).
   \end{cases}
\]
Part (i) follows using this in \eqref{e:bdiffbd}. 
Plancherel's identity yields part (ii).
\end{proof}

\begin{proof}[Proof of Proposition~\ref{p:Reps}]
From \eqref{d:calRS} and the triangle inequality, we get
\begin{equation*}
  \|\calR_\eps\|_{H^{1}} \leq \|\calQ_\eps (W_0) - \calQ_0(W_0) \|_{H^{1}} + \|(\calB_\eps\inv - \calB_0\inv ) \calQ_0(W_0) \|_{H^{1}}\,,
\end{equation*}
which are estimated respectively by Lemmas~\ref{lem:quad_diff} and~\ref{lem:bepsdiff}, using smoothness of $W_0$. 
Considering each case $\alpha \in (1,2]$ and $(2,3)$ gives the desired result. 
\end{proof}

\subsection{Derivative estimates}

Here our goal is to prove derivative estimates which will entail
the conditions \eqref{c:DF2} and \eqref{c:DF3} in Lemma~\ref{lem:IFT}. 
In fact, we seek to prove the following. 

\begin{proposition}\label{p:DF} 
(i) Given any $C_1>0$, if $0<\delta\le C_1/\kappa_2C_{H^1}$ 
then 
\[ \|D\calF(W)-D\calF(W_0)\|_{\calL(H^1)} \le \kappa_2 C_{H^1}\|W-W_0\|_{H^1} \le C_1  \]
for all $W\in H^1$ with $\|W-W_0\|_{H^1}\le\delta$.

\smallskip\noindent
(ii) For any $W\in H^1$, the operator $D\calF(W)$ is compact on $H^1$.
\end{proposition}

\begin{proposition}\label{p:DFG} 
Given any $C_2>0$ there exist positive constants $\delta$ and $\eps_0$ such that 
whenever $\eps\in(0,\eps_0)$ and $\|W-W_0\|_{H^1}\le\delta$  we have
\[
\|D\calF(W)-D\calG_\eps(W)\|_{\calL(H^1)} \le C_2\,.
\]
\end{proposition}

\begin{proof}[Proof of Proposition~\ref{p:DF}]
Define $\calN\colon H^1\to\calL(H^1)$ by
\begin{equation} \label{d:N0}
\calN(W)f=\calB_0\inv(Wf)\,.
\end{equation}
Then $D\calF = \kappa_2\calN$, and Proposition~\ref{p:DF} follows immediately
from the following lemma.
\end{proof}

\begin{lemma}\label{lem:calN} 
(i)   For any $W_1,W_2\in H^1$ we have 
\[
\|\calN(W_1)-\calN(W_2)\|_{\calL(H^1)}\le C_{H^1}\|W_1-W_2\|_{H^1} \,.
\]
(ii) For any $W\in H^1$, the operator $\calN(W)$ is compact on $H^1$.
\end{lemma}
\begin{proof} For all $V\in H^1$ we have 
    \[
    \calN(W_1)V - \calN(W_2)V  = \calB_0\inv ((W_1-W_2) V)\,.
    \]
Since $\calB_0\inv$ is nonexpansive on $H^1$ the estimate in (i) follows.

For part (ii), assume at first that $W\in C^\infty_c(\R)$.
Then the operators $\calN(W)$ and $\calN(W')$ are compact
on $L^2$ by the compactness criteria in \cite{Pego.85},  since 
the functions $b_0\inv$, $W$ and $W'$ are continuous and vanish at $\infty$.  
It follows easily that $\calN(W)$ is compact on $H^1$.  
For a general $W\in H^1$,  choose a sequence of functions $W_n\in C^\infty_c(\R)$
approximating $W$ in $H^1$.
Then part (i) implies $\calN(W)$ is approximated in $\calL(H^1)$
by the compact operators $\calN(W_n)$, hence is itself compact.
\end{proof}

\begin{proof}[Proof of Proposition~\ref{p:DFG}]
Recall  $\kappa_2=2\alpha_2\zeta_\alpha$ and 
\[
D\calF(W)V = \calB_0\inv (D\calQ_0(W)V) = \kappa_2\calN(W)V = \kappa_2 \calB_0\inv(WV)\,,
\]
\[
D\calG_\eps(W)V = \calB_\eps\inv( D\calQ_\eps(W)V + D\calZ_\eps(W)V) \,,
\]
where 
\[
D\calQ_\eps(W)V = 2\alpha_2 \sum_{m\ge1}m^{-\alpha} \calA_{m\eps}( (\calA_{m\eps}W) (\calA_{m\eps} V)\bigr) \,.
\]
Based on the multiplicative inequality for $H^1$ 
and the non-expansivity of  $\calB_0\inv$, $\calB_\eps\inv$, and  $\calA_{m\eps}$,
the proof of the following lemma is easy and is omitted. 
\begin{lemma} For all $W\in H^1$ we have
\begin{align*}
    \|\calB_\eps\inv (D\calQ_\eps(W)-D\calQ_\eps(W_0))\|_{\calL(H^1)} &\le 
\kappa_2 C_{H^1} \|W-W_0\|_{H^1}.
\end{align*} 
\end{lemma}
By this result and the bounds on $D\calZ_\eps$ in Lemma~\ref{lem:Zeps},
to prove Proposition~\ref{p:DFG} it suffices to prove that 
\begin{equation}\label{e:BQdiff1}
    \| \calB_0\inv D\calQ_0(W_0) - \calB_\eps\inv D\calQ_\eps(W_0)\|_{\calL(H^1)} 
    \to 0 \quad\text{as $\eps\to0$}.
\end{equation}

Key to our approach is the following result on operator norm convergence
of Fourier multipliers. It will be proved in the subsection to follow, 
by use of Plancherel's identity and
a proof that $b_\eps(k)\inv\to b_0(k)\inv$ uniformly in $k$.
\begin{proposition}\label{p:Bepslim}
    As $\eps\to0$ we have $\|\calB_\eps\inv - \calB_0\inv\|_{\calL(H^1)}\to 0$.
\end{proposition}
Taking this for granted at present,
since $D\calQ_\eps(W_0)$ is uniformly bounded we infer that 
to prove Proposition~\ref{p:DFG} it suffices to 
replace $\calB_\eps\inv$ by $\calB_0\inv$ in \eqref{e:BQdiff1}, i.e.,  to prove
\begin{equation}\label{e:BQdiff2}
    \| \calB_0\inv D\calQ_0(W_0) - \calB_0\inv D\calQ_\eps(W_0) \|_{\calL(H^1)} 
    \to 0 \quad\text{as $\eps\to0$}.
\end{equation}
To proceed we define operators $\calN_{\eta}$ and $\calN_0$ by 
\begin{equation}\label{d:calN}
\calN_{\eta} f = \calB_0\inv( (\calA_{\eta}W_0) f) \,,
\quad
\calN_0 f= \calB_0\inv(W_0 f).
\end{equation}
Evidently by Lemma~\ref{lem:calN} we have that 
\begin{equation}
    \|\calN_{\eta}  - \calN_0  \|_{\calL(H^1)} \le 
    C_{H^1} \|(\calA_{\eta}-I)W_0\|_{H^1} \to 0 \quad\text{as $\eta\to0$.}
\end{equation}
And we may write 
\begin{equation}
\calB_0\inv D\calQ_\eps(W_0)V = 
2\alpha_2\sum_{m\ge1} m^{-\alpha} 
\calA_{m\eps} \calN_{m\eps} \calA_{m\eps} V \,.
\end{equation}

Since $\calA_\eta$ is self-adjoint and $\calA_\eta\to I$ strongly as $\eta\to0$,
we can use the fact that $\calN_0$ is compact and 
the abstract Lemma~\ref{lem:strongnorm} below  to conclude that 
\begin{equation}
    \| \calA_\eta \calN_\eta \calA_\eta - \calN_0\|_{\calL(H^1)} \to 0 
    \quad\text{as $\eta\to0$.}
\end{equation}
Then \eqref{e:BQdiff2} follows 
from the fact that
\[
    \calB_0\inv D\calQ_0(W_0) - \calB_0\inv D\calQ_\eps(W_0) = 
    2\alpha_2 \sum_{m\ge1} m^{-\alpha} (\calN_0-\calA_{m\eps}\calN_{m\eps}\calA_{m\eps}),
\]
by using the dominated convergence theorem on the sum after taking norms.
Modulo the proofs of Proposition~\ref{p:Bepslim} and Lemma~\ref{lem:strongnorm} to come,
this completes  the proof of Proposition~\ref{p:DFG}. 
\end{proof}

\subsection{Lemmas on compactness and Fourier multipliers}
In this subsection we complete the proof of Proposition~\ref{p:DFG},
by proving the proofs of the two just-mentioned results deferred 
from the previous subsection.

\subsubsection{Compactness and operator convergence}
\begin{lemma}\label{lem:strongnorm}
   Let $\calX$ be a Banach space. Let $S,T\in \calL(\calX)$, 
   and assume $T$ is compact.
  Let $(S_n)_n$ and $(T_n)_n$ be sequences in $\calL(\calX)$, 
   and assume $\|T_n-T\|\to0$ as $n\to\infty$.  Then:
\begin{itemize}
\item[(i)] If $S_n \to S$ strongly, then $\|S_nT_n-ST\|\to0$.
\item[(ii)] If the adjoints $S_n^* \to S^*$ strongly, then $\|T_nS_n-TS\|\to0$.
\end{itemize}
\end{lemma}
\begin{proof}
To prove (i), 
suppose the claimed convergence fails.  Then there must exist a constant $c>0$ and
a sequence $(x_n)_n$ in $\calX$ such that $\|x_n\|=1$ and 
$c\le \|(S_nT_n-ST)x_n\|$
for all $n$. However, since $T$ is compact we may pass to a subsequence
(denoted the same) such that  $Tx_n\to y$ for some $y\in \calX$. Then $T_nx_n\to y$ also, while
\begin{align*}
\|(S_nT_n-ST)x_n\| &\le 
\|S_n(T_nx_n-y)\|+\|(S_n-S)y\|+\|S(y-Tx_n)\| \,.
\end{align*}
But the hypotheses ensure this tends to $0$, since $\|S_n\|$ must be uniformly bounded. This contradiction proves (i).

For (ii) we note that the compactness of $T$ on $\calX$ implies the compactness of its adjoint $T^*$ on $\calX^*$, and that
\[
\| T_nS_n-TS\| = \|S_n^*T_n^*-S^*T^*\|. 
\]
Then applying part (i) to the adjoints yields part (ii).
\end{proof}

\subsubsection{Convergence of Fourier multipliers}

Due to Plancherel's identity it is evident that 
\begin{equation}\label{d:omegab}
    \|\calB_\eps\inv - \calB_0\inv\|_{\calL(H^1)} \le 
\omega_b(\eps):= \sup_{k\in\R} |b_\eps(k)\inv - b_0(k)\inv| \,.
\end{equation}
Then Proposition~\ref{p:Bepslim} is implied by the following.
\begin{lemma}\label{lem:bepslim}
   As $\eps\to0$ we have $\omega_b(\eps)\to0$.
\end{lemma}

We prepare for the proof with some lower bounds on $b_\eps(k)$.
\begin{lemma}
Fix $h_0>3\sqrt{\zeta_{\alpha+2}/\zeta_\alpha}.$ 
Then there exist positive constants $\nu_1,\, \nu_2$ such that 
\begin{equation}
    b_\eps(k) \geq \begin{cases} 1+ \nu_1 |k|^{\alpha -1}, & |k| \leq h_0 /\eps\,, \\
    \nu_2 \eps^{1-\alpha}, & |k| > h_0 /\eps \,. \end{cases}
\end{equation}
\end{lemma}
\begin{proof} 
1. Suppose $h := \eps |k| \leq h_0$. Recall from \eqref{e:Slim} that $S_\alpha(h) \to \kappa_3$ as $h \to 0^+$. 
Since $S_\alpha$ is continuous and positive, it
attains a positive minimum on $[0,h_0]$. That is, there exists $\nu_1 > 0$ such that 
\begin{equation}
    S_\alpha(h) \geq \nu_1 \quad \text{for $0 \leq h \leq h_0$}\,, 
\end{equation}
and hence 
\begin{equation}
    b_\eps(k) \geq 1 + \nu_1 |k|^{\alpha -1} \quad\text{for $\eps|k|\le h_0$.}
\end{equation}
2. Now suppose $h=\eps|k| > h_0$. Then 
\begin{align*}
    \sum_{m=1}^{\infty}\frac{1- \sinc^2(\frac 12 mh)}{m^{\alpha}} &= 
    \zeta_\alpha - \sum_{m=1}^{\infty}\frac{4\sin^2 ( \frac 12 mh)}{h^{2}m^{\alpha+2}}
    \geq\zeta_\alpha - \frac{4}{h^2} \zeta_{\alpha+2}\,.
    \end{align*}
But since $h > h_0$, we get  
$    S_\alpha(h) \geq \alpha_1 \zeta_{\alpha} \left(1 - \frac{4}{9} \right).$
Then 
\begin{equation}
    b_\eps(k) = 1+|k|^{\alpha-1}S_\alpha(h) 
    \geq 1 + \nu_2\eps^{1-\alpha} \,, 
\end{equation}
where $\nu_2=\frac59\alpha_1\zeta_\alpha h_0^{\alpha-1}$. The lemma follows.
\end{proof}
\begin{proof}[Proof of Lemma~\ref{lem:bepslim}]
Since $b_\eps$ and $b_0$ are even, it suffices to confine attention to $k>0$.
Recall $b_\eps(0)=1=b_0(0)$, and recall from  \eqref{e:bdiffbd} that 
\begin{equation}
 \left| b_\eps^{-1}(k) - b_0^{-1}(k)\right| 
    \leq \frac{\left|S_\alpha(\eps k) - \kappa_3\right|}{\kappa_3 b_\eps(k)}\,.
\end{equation}
Let $\delta > 0$. We proceed in three steps.
1. Choose $h_\delta > 0$ such that
whenever $0 < h\leq h_\delta$, 
\begin{equation}
\left|S_\alpha(h) - \kappa_3 \right| < C_0^{-1}\delta\,.
\end{equation}
Assuming $0<\eps k\le h_\delta$,  since $b_\eps(k)\ge1$ we find from \eqref{e:bdiffbd} that 
\begin{equation}
 \left| b_\eps^{-1}(k) - b_0^{-1}(k)\right|
 \leq
    C_0 {\left|S_\alpha(\eps k) - \kappa_3\right|}
< \delta \,.
\end{equation}
2. Next, assume $h_\delta \leq \eps k \leq h_0$, where $h_0$ was introduced in the previous lemma. 
By \eqref{e:Salphabd} we get 
\[
\left|S_\alpha(\eps k) - \kappa_3 \right| < 
h_\delta^{1-\alpha}\zeta_\alpha+\kappa_3 =: C_1 .
\]
Then from the previous lemma, it follows
\begin{align}
    C_0 \frac{\left|S_\alpha (\eps k) - \kappa_3\right|}{b_\eps(k)}
    &\leq \frac{C_0C_1}{1+ \nu_1 k^{\alpha-1}}  \leq C_0C_1 h_\delta^{1-\alpha}\eps^{\alpha -1}\,.
\end{align}
3. Lastly, assume $h_0 \leq \eps k < \infty$. 
With the second bound from the previous lemma, we get
\begin{align}
    C_0 \frac{\left|S_\alpha (\eps k) - \kappa_3\right|}{b_\eps(k)}
    \leq  \frac{C_0C_1}{\nu_2} \eps^{\alpha-1}\,.
\end{align}

Using the inequalities above, we see there exists $\eps_0>0$ (depending on $\delta$) such that 
for all $\eps \in (0,\eps_0)$ and all $k \in (0,\infty)$,
\begin{equation}
    \left|b_\eps^{-1}(k) - b_0^{-1}(k)\right| < \delta.
\end{equation}
This finishes the proof of the lemma.
\end{proof}

\subsection{Existence proof}\label{ss:exist}
We are now in a position to prove the part of Theorem~\ref{t.main} 
concerning the existence and local uniqueness of solitary wave profiles, 
by invoking Lemma~\ref{lem:IFT} to obtain fixed points of \eqref{d:Geps} for small $\eps>0$.

\begin{proof}[Proof of Theorem \ref{t.main} (existence and local uniqueness)] 
Let $E=\Heven^1$ and suppose $W_0\in E$ is a non-degenerate solution of \eqref{e:W0eq}.
Then $u_0=W_0$ is a fixed point of $F=\calF$ in $E$. The operator $\calL_0=I-D\calF(W_0)$ on $E$ is
Fredholm due to Proposition~\ref{p:DF}(ii) and has trivial kernel in $E$, hence is invertible.
Let $C_0=\|\calL_0\inv\|_{\calL(E)}$ and choose positive constants $\theta$, $C_1$ and $C_2$
such that \eqref{c:012} holds, i.e., $C_0(C_1+C_2)<\theta<1$.  

Let $R>C_{H^1}\|W_0\|_{H^1}$ and let 
\[
B=\{f\in E: C_{H^1}\|f\|_{H^1}\le R\}.
\]
By applying Corollary~\ref{cor:res} and Propositions~\ref{p:Fanalytic}, \ref{p:DF} and \ref{p:DFG},
we can find positive constants $\delta$ and $\eps_0$ sufficiently small, 
such that whenever $0<\eps<\eps_0$,
then:
\begin{itemize} \renewcommand{\itemsep}{0pt}
\item[(i)] $\calF$ and $\calG_\eps$ are analytic on $B$,
\item[(ii)] the residual bound \eqref{c:DF1} holds, and 
\item[(iii)] whenever $\|u-u_0\|_E\le\delta$ we have 
$u\in B$ and estimates \eqref{c:DF2} and \eqref{c:DF3} hold.
\end{itemize}
Then with $G=\calG_\eps$, Lemma~\ref{lem:IFT} applies and we conclude that
for every $\eps\in(0,\eps_0)$, $\calG_\eps$ has a unique fixed point 
$W=W_\eps\in \Heven^1$ satisfying $\|W-W_0\|_{H^1}\le\delta$.
Moreover there is a constant $C$ independent of $\eps$ such that 
\begin{equation}\label{e:final_bound}
\|W_\eps-W_0\|_{H^1} \le C\|\calF(W_0)-\calG_\eps(W_0)\|_{H^1}  \le
   \begin{cases} 
     C\eps^{\alpha-1} & \alpha\in (1,2],\\
     C \eps^{3-\alpha} & \alpha\in(2,3).
   \end{cases}
\end{equation}
the last bound being due to Corollary~\ref{cor:res}.
\end{proof}

\subsection{The Calogero-Moser case}\label{ss:CM}
In the case $\alpha=2$ that corresponds to an infinite Calogero-Moser lattice, we recall from \cite[Theorem~1.1]{IP24} that 
traveling waves in the form $x_j(t) = j-\vp(j-ct)$ exist for any $c>c_\alpha=\pi$, where the function $\vp=\vp(z)$ takes values in $(-\frac12,\frac12)$ and is determined for all $z\in \R$
by the implicit equation
\begin{equation}
(c^2 - \pi^2)(z- \varphi) = \pi \tan{\pi \varphi}\,.
\label{e:explicit_eq}
\end{equation}
We seek to relate the velocity profile $v_c(z)=c\vp'(z)$ to the fixed point $W_\eps$ 
provided by Theorem~\ref{t.main}
with $W_0$ taken to be the ground state solution of \eqref{e:W0eq}
(known to be non-degenerate by \cite{Frank.Lenzmann.13}).  
Here, equation \eqref{e:W0eq} takes the form
\begin{equation}\label{e:W0_CM}
     W + \pi |D| W = \frac{1}{2} (2\pi W)^2 \,,
\end{equation}
since  $\kappa_2=4\pi^2$ and $\kappa_3=\pi$ when $\alpha=2$. Using that $f(z)=i/(z+i\pi)$ satisfies $f'=if^2$, one can check that
a solution of \eqref{e:W0_CM} is given by $\re f(\xx)/\pi.$
By the classical uniqueness result of Amick and Toland~\cite{AmickToland1991b}, this is the only solution of \eqref{e:W0_CM} in $\Heven^1$.
Therefore,
\[
W_0(\xx) = \frac{1}{\xx^2+\pi^2} \,.
\]

\begin{corollary}
If $\alpha =2$ and $\eps > 0$ is sufficiently small, then the fixed point $W_\eps$ from Theorem~\ref{t.main} with 
$c^2 = \pi^2 + \eps$ 
precisely satisfies $\vp'(z) = \eps W_\eps(\eps z)$, where $\vp$ satisfies \eqref{e:explicit_eq}.
\end{corollary}
\begin{proof}
Let $c^2 = \pi^2 + \eps$. Define $\psi_\eps$ by $\vp'(z)=\eps\psi_\eps(\eps z)$  where $\varphi_\eps$ satisfies \eqref{e:explicit_eq} with $z = j - ct$. 
Since $\vp$ determines a solitary wave for \eqref{e:sys1} by \cite[Theorem~1.1]{IP24}, 
and by the discussion in Section~\ref{s:equations} above, 
$\psi_\eps$ must satisfy the fixed point equation~\eqref{d:Geps}. 
From Theorem~\ref{t.main}, $W_\eps$ is the unique fixed point of \eqref{d:Geps} satisfying $\|W-W_0\|_{H^1}\le\delta$.
Thus to show $\psi_\eps=W_\eps$ it remains to show $\|\psi_\eps-W_0\|_{H^1}\le \delta$ for small enough $\eps>0$.

Now, by differentiation of \eqref{e:explicit_eq}, one derives that 
\[
\psi_\eps(\xx) = \frac{1}{(\xx-\eps\vp)^2+\pi^2+\eps} \,.
\]
 Then it is straightforward to check that $\psi_\eps$ converges to $W_0$ in $H^{1}$ as $\eps \to 0$ due to the boundedness of $\psi_\eps'$ and $\psi_\eps''$. 
 By the local uniqueness in Theorem~\ref{t.main}, the proof is complete.
 \end{proof}

\section{Positivity and regularity}\label{s:positive}

In this section we establish the positivity and regularity properties 
of the velocity profiles that were
stated in Theorem~\ref{t.main}.

\subsection{Positivity} 
First we remark on reasons why any solution $W_0\in H^1$ of \eqref{e:W0eq} is positive.
As we have pointed out, the Green's function for $\calB_0=I+\kappa_3 |D|^\mu$ is positive.
This follows by scaling from \cite[Lemma~A.4]{Frank.Lenzmann.13}. Alternatively, it can be proved 
by invoking Kato's formula \cite{Kato.60} to show that for any $\lambda>0$ and $s\in(0,1)$,
\begin{equation}
(\lambda + |D|^{2s})\inv = \frac{\sin \pi s}\pi \int_0^\infty 
\frac{t^s}{\lambda^2+ 2\lambda t^s\cos(\pi s)+t^{2s}} (t I-\Delta)\inv\,dt \,,
\end{equation}
and using the positivity of the Green's function for $tI-\Delta$,
which in dimension one 
is $e^{-\sqrt t|x|}/2\sqrt{t}$. 
Curiously, we can get a third proof by taking the limit $\eps\to0$ in the next lemma,
which we will use to study \eqref{d:Geps}.

\begin{lemma}\label{lem:positive1} Let $f\in \Heven^1$.
If $f$ is positive (resp. unimodal) then $\calB_\eps\inv f$ is positive (resp. unimodal).
\end{lemma}
\begin{proof} The proof is essentially similar to one provided in \cite[Cor.~2.7]{HML16} 
for the corresponding operator in the case of finite-range interactions.
    From \eqref{d:beps} we may write 
        $b_\eps(k) = 1 + \alpha_1\zeta_\alpha \eps^{-\mu}(1-j_\eps(k))$, where
    \begin{equation}
        j_\eps(k) = \zeta_\alpha\inv \sum_{m\ge1} m^{-\alpha}\sinc^2(\tfrac12km\eps).
    \end{equation}
Then $j_\eps$ is even, takes values in $[0,1]$, and is the symbol of the Fourier multiplier
    \begin{equation}
        \calJ_\eps = 
         \zeta_\alpha\inv \sum_{m\ge1} m^{-\alpha}\calA_{m\eps}^2 \,.
    \end{equation}
The operator norm $\|\calJ_\eps\|_{\calL(H^1)}\le1$, hence by Neumann series expansion,
\begin{equation}
    \calB_\eps\inv = \frac{\eps^\mu}{\alpha_1\zeta_\alpha} 
    \sum_{n=0}^\infty \frac{\calJ_\eps^n}{(1+\eps^\mu/\alpha_1\zeta_\alpha)^{n+1}} \,,
\end{equation}
and the series converges in operator norm. Suppose $f\in\Heven^1$ and $f$ is positive
(resp. unimodal). Since the same is true for $\calA_{m\eps}f$, for all $f$, we infer
that $\calJ_\eps f$ is positive (resp. unimodal).  By induction, the same is true for
$\calJ_\eps^n f$, for all $n\ge1$. It follows that $\calB_\eps\inv f$ is positive (resp. unimodal) as well.
\end{proof}

\begin{lemma}\label{lem:positive2}
    Let $f\in\Heven^1$ with $C_{H^1}\|f\|_{H^1}<1.$
Then $\calG_\eps(f)$ is positive.
Moreover, if $f$ is unimodal, then $\calG_\eps(f)$ is unimodal.
\end{lemma}

\begin{proof} From the definitions \eqref{d:Qeps}--\eqref{d:Zeps}, we find
\begin{equation}
\eps^{2\mu}(\calQ_\eps(f)+\calZ_\eps(f)) =  
\sum_{m=1}^\infty \frac1{m^\alpha} \calA_{m\eps} Z_2(\eps^\mu \calA_{m\eps} f)\,,
\end{equation}
where $Z_2(r) = \alpha(1-r)^{-\alpha-1} -\alpha-\alpha_1 r$. Because $Z_2$ is strictly convex
with $Z_2(0)=Z_2'(0)=0$ we have $Z_2(r)>0$ for $0<|r|<1$. 
Because $\calA_{m\eps}$ preserves positivity, 
by Lemma~\ref{lem:positive1} it follows $\calG_\eps(f)$ is positive.
A similar argument applies to the unimodality statement.
\end{proof}

\begin{proof}[Proof of Theorem \ref{t.main}(positivity)]
The positivity of the fixed points $W_\eps$ of $\calG_\eps$ proved to exist in Section~\ref{ss:exist}
follows immediately from Lemma~\ref{lem:positive2}.
\end{proof}

\noindent
{\it Remarks on unimodality.}
Regarding the question of whether $W_\eps$ is unimodal if $W_0$ is, 
we can only reiterate what was said on this subject by 
Herrmann and Mikikits-Leitner~\cite[p.~2065]{HML16}. 
Unimodality would follow, if, starting from $W_0$, 
one could show that $W_\eps$ arose as a fixed-point limit
of a suitable variant of the (unstable) iteration scheme
\[
W \mapsto \calG_\eps(W) = \calB_\eps\inv(\calQ_\eps(W)+\calZ_\eps(W)),
\]
Perhaps for this one could use Petviashvili iteration \cite{petviashvili1976equation,PelinovskyEA.2016}, say, or compactness arguments
similar to those Herrmann used in \cite{Herrmann.10} for nearest-neighbor forces.
Also see \cite{HerrmannMatthies.20}.
The analysis involved is outside the scope of the present paper, however. 

\subsection{Regularity of velocity}\label{ss:regularityW}
We will prove that the fixed points $W_\eps$ are in $H^\infty$
by a bootstrap argument based on equation \eqref{e:nevp}. 
We provide details since the terms in the infinite series depend on $m$ (though weakly).

\begin{proof}[Proof of Theorem~\ref{t.main}(regularity)]
Throughout the proof we keep $\eps\in(0,\eps_0)$ fixed and write 
$W=W_\eps$ and $a_m = \eps^\mu \calA_{m\eps}W$.  By the choice of $\eps_0$
in the existence proof we have that $C_{H^1}\|a_m\|_{H^1}\le \eps^\mu R\le \rho$
where $\rho<1$.  
 We note that for every $m\ge 1$ and every $k\ge1$, 
we have $Z^{(k)}\circ a_m-Z^{(k)}(0)\in H^1$  with 
\begin{equation}\label{e:Zkbd}
\| Z^{(k)}\circ a_m-Z^{(k)}(0)\|_{H^1} \le Z^{(k)}(\rho) - Z^{(k)}(0) \,,
\end{equation}
due to the fact that the Maclaurin series for $Z(r)$
has positive coefficients and unit radius of convergence. 

We will prove by induction that for every integer $n\ge0$, $W\in H^{n+1}$ and 
\begin{equation} \label{e:induct}
\eps^\mu c^2 W^{(n)} = 
\sum_{m=1}^\infty m^{-\alpha}
\calA_{m\eps}(Z_1\circ a_m)^{(n)} \,,
\end{equation}
where $Z_1(r)=Z(r)-\alpha$, with the series converging in $H^1$.
This holds for $n=0$, since $W\in H^1$
and  \eqref{e:nevp} holds in $H^1$.

Now fix $n\in \N$ and suppose $W\in H^{n+1}$ with \eqref{e:induct} holding in $H^1$.
Then $W\in C^n$,  $(Z_1\circ a_m)^{(n)}=(Z\circ a_m)^{(n)}$,
and by the Fa\`a di Bruno formula, 
\begin{equation}\label{e:Faa}
(Z\circ a_m)^{(n)} = \sum_{\boldsymbol{k}\in \Lambda_n}
\binom{n}{\boldsymbol{k}}
(Z^{(|\boldsymbol{k}|)}\circ a_m)\cdot 
\prod_{j=1}^n \left(\frac{a_m^{(j)}}{j!}\right)^{k_j}  \,,
\end{equation}
where 
\[
\Lambda_n = \left\{{\boldsymbol{k}=(k_1,\ldots,k_n)}\in \N^n: \sum_{j=1}^n j k_j=n\right\},
\quad 
\binom{n}{\boldsymbol{k}} = \frac{n!}{k_1!\cdots k_n!}, \quad 
\]
and $|\boldsymbol{k}|= k_1+\ldots+k_n$.
From \eqref{e:Faa} and \eqref{e:Zkbd},
it follows easily that $(Z\circ a_m)^{(n)}$ is 
bounded in $H^1$ uniformly in $m$,
by writing 
\[
Z^{(|\boldsymbol{k}|)}\circ a_m 
= Z^{(|\boldsymbol{k}|)}\circ a_m  - Z^{(|\boldsymbol{k}|)}(0)
+ Z^{(|\boldsymbol{k}|)}(0)\,,
\]
and using the Banach algebra property of $H^1$
together with the fact that 
$\|a_m^{(j)}|_{H^1}\le \|\eps^\mu W^{(j)}\|_{H^1}$ for all $m$.

The map $\calA_{m\eps}$ is bounded from $H^1$ into $H^2$ with bound
independent of $m$. (The bound depends on $\eps$ but it does not matter here.)
We infer therefore that the series \eqref{e:induct} converges in $H^2$.
Hence $W\in H^{n+2}$ and \eqref{e:induct} holds in $H^1$ with $n$ replaced by $n+1$.
This completes the induction step, and finishes the proof that $W\in H^\infty$.
\end{proof}

\subsection{Regularity in wave speed}
\label{ss:regularity_c}

In this subsection we prove the part of Theorem~\ref{t.main} 
stating that the unscaled velocity profile is analytic as a function of wave speed.
Similar to what was done in \cite{Friesecke.Pego.99}, we look at a fixed scaling, 
and apply the analytic implicit function theorem in complexified Banach spaces,
as provided by Berger~\cite[Section 3.3B]{Berger}.

Let $W_\eps$ be the wave profiles provided by the existence proof in Subsection~\ref{ss:exist} for $0<\eps<\eps_0$.
Fixing some such $\eps$, define
\begin{equation}\label{d:Wepsbeta}
        W_{\eps,\beta}(\xx) := \eta^{\mu}W_{\eta\eps} (\eta \xx)\, ,\quad \beta = \eta^\mu\,,
\end{equation}
whenever $0<\eta\eps<\eps_0$.
This function is related to the unscaled velocity profiles described in \eqref{d:vc} by
\begin{equation}\label{e:vc_beta}
    v_c(z) = c\eps^\mu W_{\eps,\beta}(\eps z) \quad\text{with}\quad  c^2 = c_\alpha^2 + \beta \eps^\mu\,.
\end{equation}
Thus, to study the regularity of $v_c$ as a function of $c$, it suffices to fix $\eps$ and study
$W_{\eps,\beta}$ as a function of $\beta$ in an interval around $\beta=1$.
Define 
\begin{equation}\label{e:tw_beta}
        \calB_{\eps,\beta} = \beta I  + \alpha_1 \sum_{m \geq 1} \frac{\eps^{-\mu}}{m^{\alpha}} (I - \calA_{m\eps}^2)\,.
\end{equation}

\begin{proposition}\label{prop:betaAnalyt}
        (1) For $0<\eta\eps <\eps_0$, $W_{\eps,\beta}$ satisfies the traveling wave equation
\begin{equation}\label{e:beta_evp3}
  \calB_{\eps,\beta} V = \calQ_\eps (V) + \calZ_\eps(V)\,.  
\end{equation}
(2) Moreover, there exists an interval $(\beta_-,\beta_+)$, which contains $1$ and depends upon $\eps$, 
on which the map  $\beta \mapsto W_{\eps,\beta} \in \Heven^{1}$ is analytic.  
\end{proposition}
\begin{proof}[Proof of (1)]
Using the scaling formulas in the following lemma, we get that $W_{\eps,\beta }$ 
solves the traveling wave equation \eqref{e:beta_evp3} after setting $\hat\eps=\eta\eps$ and multiplying
\begin{equation}
  (\calB_{\hat{\eps}} W_{\hat{\eps}})(\eta \xx) = (\calQ_{\hat{\eps}}(W_{\hat{\eps}}) + \calZ_{\hat{\eps}} (W_{\hat{\eps}}))(\eta \xx)
\end{equation}
by $\eta^{2\mu}$.  
\end{proof}
\begin{lemma}
We have
\begin{align}
  \eta^{\mu}(\calA_{m\hat{\eps}}W_{\hat\eps})(\eta \xx) &= (\calA_{m\eps} W_{\eps,\beta})(\xx)\,,\\
  \eta^{k\mu} \calA_{m\hat{\eps}}[ (\calA_{m\hat{\eps}} W_{\hat{\eps}})]^{k}(\eta \xx) &= \calA_{m\eps} [ (\calA_{m\eps} W_{\eps,\beta})]^{k}(\xx)\,,\\
  \calA_{m\hat{\eps}} Z_3 (\hat{\eps}^{\mu}(\calA_{m\hat{\eps}}W_{\hat{\eps}}))(\eta \xx) &= \calA_{m\eps} Z_3 (\eps^{\mu}(\calA_{m\eps} W_{\eps,\beta}))(\xx)\,.
\end{align}
\end{lemma}

\begin{proof}
Through the change of variables $z=\eta y$, we get
\begin{align*}
         \eta^\mu (\calA_{m\hat{\eps}}W_{\hat{\eps}} )(\eta \xx) 
    &= \frac{1}{m\hat{\eps}} \int_{-m\hat{\eps}/2}^{m\hat{\eps}/2} \eta^\mu W_{\hat{\eps}}(\eta \xx + z)\, dz \\
    &= \frac{1}{m\eps} \int_{-m\eps/2}^{m\eps/2} W_{\eps,\beta}(\xx+y)\, dy
    = (\calA_{m\eps}W_{\eps,\beta})(\xx)\,.
\end{align*}
Similarly, for all $k\ge1$,
\begin{align*}
      \eta^{k\mu} \calA_{m\hat{\eps}}[(\calA_{m\hat{\eps}} W_{\hat{\eps}})]^{k}(\eta \xx) 
      &= \frac{1}{m\hat{\eps}} \int_{-m\hat{\eps}/2}^{m\hat{\eps}/2} 
      \left[\eta^\mu (\calA_{m\hat{\eps}} W_{\hat{\eps}}) (\eta \xx + z)\right]^{k}\, dz\\
      &= \frac{1}{m\eps} \int_{-m\eps/2}^{m\eps /2} [ (\calA_{m\eps} W_{\eps,\beta})(\xx+y)]^{k}\, dy\\
      &= \calA_{m\eps}[(\calA_{m\eps} W_{\eps,\beta})]^{k}(\xx)\,.
\end{align*}
Finally, 
\begin{align*}
     \calA_{m\hat{\eps}} Z_3 (\hat{\eps}^\mu (\calA_{m\hat{\eps}} W_{\hat{\eps}}))(\eta \xx) 
     &= \sum_{k \geq 3} \alpha_k \calA_{m\hat{\eps}} [ \eps^\mu\eta^\mu (\calA_{m\hat{\eps}} W_{\hat{\eps}} )]^{k}(\eta \xx)\\
    &= \sum_{k \geq 3} \alpha_k \calA_{m\eps} [\eps^\mu(\calA_{m\eps} W_{\eps,\beta})]^{k}(\xx)\\
    &= \calA_{m\eps} Z_3(\eps^{\mu}(\calA_{m\eps}W_{\eps,\beta}))(\xx)\,.
    \qedhere
        \end{align*}
\end{proof}

\begin{lemma}\label{lem:Bepsbeta}
The mapping $ \beta \mapsto \calB_{\eps,\beta}^{-1} \in \calL(H^{1})$ is analytic on $(0,\infty)$.       
\end{lemma}
\begin{proof}
First, by linearity of the inverse Fourier transform, it suffices to show that $\beta \mapsto b_{\eps,\beta}^{-1} \in L^{\infty}$ is analytic. Let $\beta_0 \in (0,\infty)$. Naming $f(k) = |k|^{\alpha -1} S_\alpha(\eps|k|)$, we get
\begin{align*}
    b_{\eps,\beta}^{-1}(k) &= \frac{1}{\beta  + f(k)}
    = \frac{1}{\beta_0 + f(k)} \frac{1}{\frac{\beta - \beta_0}{\beta_0 + f(k)} + 1}\\
    &= \sum_{n\geq 0} \frac{(-1)^{n}}{(\beta_0 + f(k))^{n+1}} (\beta - \beta_0)^{n}\, ,
\end{align*}
granted that $|\beta - \beta_0| \leq \beta_0$, since $f(k)\ge0$. Hence, the mapping is analytic. 
\end{proof}

\begin{lemma}
Let $\eps > 0$. The operator $\tilde\calG_\eps:(0,\infty) \times (\tilde{B}_R\cap \Heven^{1}) \to  \Heven^{1}$ defined by 
\begin{equation*}
                \tilde\calG_\eps (\beta, V) = \calB_{\eps,\beta}^{-1}(\calQ_{\eps}(V) + \calZ_\eps(V))
\end{equation*}
is analytic, jointly in $\beta$ and $V$.
\end{lemma}

We omit the proof of this lemma, as it is straightforward to justify local convergence of power series expansions
given the results of Lemmas~\ref{lem:Zeps} and \ref{lem:Bepsbeta}.


\begin{proof}[Proof of Proposition~\ref{prop:betaAnalyt} (2)]
For small enough $\eps$, it follows from estimates in Theorem~\ref{t.main} that  
$I - D\calG_\eps(W_\eps)$ is invertible. This is the partial derivative 
of the function $f(\beta,V):=V-\tilde\calG(\beta,V)$ with respect to $V$, 
at the point $(1,W_\eps)$  where $f$ vanishes.
Using the joint analyticity to develop a power series expansion at the point $(1,W_\eps)$, 
we can extend $f$ to be analytic in a ball around $(1,W_\eps)$ in the complexification 
of the real Hilbert space $\R\times\Heven^1$. The Frech\`et derivative $D_Vf$ at this point
is the natural extension of the real operator $I-D\calG_\eps(W_\eps)$ and remains invertible.
We can deduce then from the analytic implicit function theorem (see \cite[Theorem 3.3.2]{Berger}) that 
for some interval $(\beta_-,\beta_+)$ containing $1$, there exists an analytic mapping $\beta \mapsto \tilde W_{\eps,\beta}$ 
taking values in $\Heven^1$ (complexified) such that $\tilde W_{\eps,\beta}$ is a solution of \eqref{e:beta_evp3}. 
But by local uniqueness, we deduce that $\tilde W_{\eps,\beta} = W_{\eps,\beta}$.
\end{proof}

\section{Hamiltonian energy and wave speed}\label{s:hamiltonian}
Let us study the behavior of the Hamiltonian as a function of wave speed.
We have not yet written a Hamiltonian for system~\eqref{e:sys1},
due to complications over convergence of the double sums that appear. 
To proceed we describe a potential function related to the 
force function $Z(r)=\alpha(1-r)^{-\alpha-1}=\sum_{k=0}^\infty \alpha_k r^k$,
defined so that $\vp_2(0)=0$ and $\vp_2'(r)=Z(r)-\alpha$, 
whence
\begin{equation}\label{d:phi}
\vp_2(r) = (1-r)^{-\alpha} - 1 -\alpha r =
\sum_{k=2}^\infty \frac{\alpha_{k-1} r^k}{k}
\,. 
\end{equation}
The lattice Hamiltonian, kinetic, and potential energies are
regarded as functions of the particle positions $x_j$ and momenta $p_j=\dot x_j$
and are given by 
\begin{align}
\calH = \calK + \calP, 
\quad
\calK = \sum_{j\in\Z} \tfrac12 p_j^2 \,,  \quad
\calP = \sum_{j\in\Z} \sum_{m=1}^\infty m^{-\alpha}\vp_2(r_{j+m,j}) \,,
\end{align}
where the quantities $r_{k,j}$, representing normalized relative compressions,
are defined via 
\[
1- r_{k,j} = \frac{x_k-x_j}{k-j}\,.
\]
In particular, note
\[
(x_{j+m}-x_j)^{-\alpha} = m^{-\alpha}(1-r_{j+m,j})^{-\alpha} \,.
\]
It is straightforward to check that the canonical Hamiltonian equations
for $\calH$ yield \eqref{e:sys1} and that $\calH$ is finite and
constant in time for solitary wave solutions.
The main result in this section is the following result which links the value
of the Hamiltonian to an approximation of
the squared $L^2$ norm of the unscaled velocity profile 
\[
\dot x_j(t) = c\eps^\mu W_\eps(\eps z), \quad z =j-ct.
\]
\begin{theorem}\label{t.ham}
For the Hamiltonian $\calH$ evaluated along the family of solitary waves
given by Theorem~\ref{t.main}, we have  
\begin{align}
\label{e:Hc}
\calH &= 
\eps^{2\mu-1}\left(\int_\R c_\alpha^2 W_0(\xx)^2\,d\xx + O(\eps^\gamma)\right)\,,\\ 
\label{e:dHdeps}
\frac{d\calH}{d\eps} &=
(2\mu-1)\eps^{2\mu-2}\left(\int_\R c_\alpha^2 W_0(\xx)^2\,d\xx + o(1)\right)\,.
\end{align}
Thus for $\alpha\ne\frac32$, $\sgn d\calH/dc$ agrees with $\sgn(\alpha-\frac32)$
for small enough $\eps>0$.
\end{theorem}

For the Calogero-Moser case $\alpha=2$,
when the solitary waves are determined through \eqref{e:explicit_eq}
by \cite[Theorem~1.1]{IP24}, we can be more explicit.
\begin{theorem}\label{t.CMH}
   In the case $\alpha=2$, for the solitary waves determined by \eqref{e:explicit_eq},
   for every wave speed $c>\pi$ we have 
   \[
   \calH = \tfrac12(c^2-\pi^2)\,.
   \]
\end{theorem}

To study the Hamiltonian on solitary waves, we write the waves 
provided by Theorem~\ref{t.main} in the form
\[
x_j(t) = j- q(j-ct), \quad \dot x_j(t) = -p(j-ct) \,,
\]
temporarily suppressing dependence on wave speed 
(and with apologies for the sign reversals but noting $\D_t q = p = -cq'$).
With $z=j-ct$ we then have that 
\[
r_{j+m,j} =  \delta^+_m q(z) := \frac{q(z+m)-q(z)}{m} \,.
\]
As in \cite{Friesecke.Pego.99},
we average the Hamiltonian over a time interval $[0,1/c]$ to reduce
it to an integral over $\R$. Because $dz=-c\,dt$, we find the expressions
\begin{align}
\calH = c \int_0^{1/c} \calH \,dt 
 & = c \int_{-\infty}^\infty \Bigl( \tfrac12 p(z)^2 + 
\sum_{m=1}^\infty m^{-\alpha}\vp_2( \delta^+_mq(z) ) \Bigr)\,dt
\nonumber
\\ & = \int_{-\infty}^\infty \Bigl( \tfrac12 p(z)^2 + 
\sum_{m=1}^\infty m^{-\alpha}\vp_2( \delta^+_mq(z) ) \Bigr)\,dz \,.
\label{e:hamreduce}
\end{align}

Although the lattice system~\eqref{e:sys1} does not admit a continuous
spatial symmetry, we note that traveling wave profiles are nevertheless
formally critical points of an ``energy-momentum'' functional 
$\calH+c\calI$, where $\calI$ is the Noether functional 
associated with the (Lagrangian) translation invariance of \eqref{e:hamreduce} and is given by 
\begin{equation}\label{d:calI}
\calI = \int_{-\infty}^\infty p(z)\D_z q(z)\,dz.
\end{equation}
Indeed, setting to zero the variations of $c\calI+\calH$ with respect
to $p$ and $q$ yields 
\begin{align}
0 &= c\D_z q(z)+p(z) \,,
\label{e:qp1}
\\
0 &= -c\D_z p + \sum_{m=1}^\infty
m^{-\alpha-1}\Bigl( Z(\delta^+_m q(z))- Z(\delta^+_m q(z-m)) \Bigr) \,,
\end{align}
which are the correct equations for solitary wave profiles.
In other words, on solitary wave profiles we have
$\delta\calH=-c\delta\calI$, a fact which will simplify 
a monotonicity calculation below.
(This functional $\calI$ differs from the physical momentum $\sum_j p_j$ generated by the
translational symmetry $x_j\mapsto x_j+h$, however.)

\begin{proof}[Proof of Theorem~\ref{t.ham}]
With these relations established, let us now insert the scaled
form of solitary wave profiles provided by our main theorem. 
We indicate by subscript the dependence of the profile tuple
$u_c=(q_c,p_c)$ upon wave speed $c$.
In particular, our ansatz \eqref{d:xj} and 
the relation $U'=W_\eps$ yields
\begin{equation}
q(z) = \eps^\nu U(\xx)\,,\quad 
\quad
p(z) = -c\D_zq(z)=
-c\eps^\mu W_\eps(\xx) \,,
\end{equation}
with $\xx=\eps(j-ct)=\eps z$. Then,
we have the relation
\[
\delta^+_m q(z) = \eps^\mu\calA_{m\eps}W_\eps(\xx+\tfrac12m\eps) \,,
\]
and, using the facts that $\vp_2(r)= O(r^2)$ and $W_\eps\in H^1$, we obtain
\begin{align}
\calH(u_c) &= 
\eps^{2\mu-1}  \int_\R \Bigl( \frac12 c^2 W_\eps(\xx)^2 + 
\sum_{m=1}^\infty m^{-\alpha}\eps^{-2\mu}\vp_2(\eps^\mu\calA_{m\eps}W_\eps(\xx) ) \Bigr)\,d\xx \,,
\\
 \calI(u_c) &= -\eps^{2\mu-1}  \int_\R c\, W_\eps(\xx)^2\,d\xx \,.
\label{e:calIuc}
\end{align}
Write $\vp_3(r)=\sum_{k=3}^\infty \alpha_{k-1}r^k/k$,
so that $\vp_2(r) = \frac12 \alpha_1 r^2 + \vp_3(r)$, and define
\begin{align}
\label{d:P2}
\calP_{2,\eps}(W) &= \frac{\alpha_1}{2} \int_\R  
\sum_{m=1}^\infty m^{-\alpha} (\calA_{m\eps}W(\xx))^2 \,d\xx\,,
\\
\label{d:P3}
\calP_{3,\eps} &= 
\int_\R \sum_{m=1}^\infty m^{-\alpha}\eps^{-2\mu}\vp_3(\eps^\mu\calA_{m\eps}W_\eps(\xx) ) \,d\xx \,.
\end{align}
In terms of these expressions we have 
\begin{equation}
\label{e:H_Weps}
\calH(u_c) = \eps^{2\mu-1}\left(
\int_\R \frac12c^2 W_\eps^2\,d\xx 
+ \calP_{2,\eps}(W_\eps) 
+ \calP_{3,\eps}\right) \,.
\end{equation}
\begin{proposition}\label{p:Ham}
As $\eps\to 0$ we have 
\[
\calH(u_c) = \eps^{2\mu-1} \left( \int_\R c_\alpha^2 W_0^2\,d\xx + O(\eps^\gamma) \right)\,.
\]
\end{proposition}

\begin{proof}
1. By Theorem~\ref{t.main} we have that $\|W_\eps-W_0\|_{H^1}=O(\eps^\gamma)$, hence
\[
\left|\int_\R W_\eps^2\,d\xx - \int_\R W_0^2 \,d\xx \right|\le 
C \|W_\eps-W_0\|_{L^2} \le C \eps^\gamma.
\]
2. Noting that 
\[
\left| \int_\R (A_{m\eps}W_\eps)^2\,d\xx - 
\int_\R (A_{m\eps}W_0)^2\,d\xx  \right| 
\le C\|\calA_{m\eps}(W_\eps-W_0)\|_{L^2} \le C\eps^\gamma,
\]
straightforward estimates imply
\[
|\calP_{2,\eps}(W_\eps)-\calP_{2,\eps}(W_0)| \le C\eps^\gamma.
\]
Furthermore, by using \eqref{e:Aeta_est} and the regularity of $W_0$ we get
\[
\left|\int_\R (A_{m\eps}W_0)^2\,d\xx -\int_\R W_0^2\,d\xx \right|  
\le C\|\calA_{m\eps}W_0-W_0\|_{L^2} \le C (m\eps)^2 \,,
\]
so by splitting the sum in \eqref{d:P2}
just as in the proof of Lemma~\ref{lem:quad_diff}, we find
\begin{equation}\label{e:P2diff}
\left|\calP_{2,\eps}(W_0) - \tfrac12\alpha_1\zeta_\alpha \int_\R W_0^2\,d\xx\right|
\le C\eps^\mu \,.
\end{equation}

3. By arguments nearly identical to those that establish the estimates for
$\calZ_3$ in Lemma~\ref{lem:Zeps}, we find that 
\[
|\calP_{3,\eps}|\le C\eps^\mu.
\]
4. Recalling that $c_\alpha^2 = \alpha_1\zeta_\alpha$ and $c^2=c_\alpha^2+\eps^\mu$
and $\mu\ge \gamma$, the proof is finished by using the estimates in steps 1-3
to estimate the terms in \eqref{e:H_Weps}.
\end{proof}

This proposition establishes \eqref{e:Hc}, and it remains
to discuss the monotonicity of solitary-wave energy as a function of wave speed.
Define 
\[
W_{0,\beta}(\xx) = \eta^{\mu}W_{0}(\eta \xx) \quad\text{where}\quad
\eta = \beta^{1/\mu}\,.
\]
Through scaling, we find $W_{0,\beta}$ to be a solution of the limiting equation 
\begin{equation}\label{e:beta_W0}
\calB_{0,\beta}V = \calQ_0(V), \quad \calB_{0,\beta} = \beta I + \kappa_3|D|^\mu \,, 
\end{equation}
which reduces to \eqref{e:W0eq-intro} when $\beta=1$.
\begin{lemma}
We have 
\begin{equation}
                \int_\mathbb{R} 2W_0 \frac{\partial}{\partial \beta} W_{0,\beta} \Big|_{\beta = 1} = \frac{2\mu - 1}{\mu} \int_\mathbb{R} W_0^2\,,
\end{equation}
and that as $\eps \to 0$,
\begin{equation}
  \left\| \frac{\partial}{\partial \beta} (W_{\eps,\beta} - W_{0,\beta})\Big|_{\beta = 1} \right\|_{H^{1}} \to 0\,.
\end{equation}
\end{lemma}
\begin{proof}
1. We have 
\begin{align}
  \int_\mathbb{R} W_{0,\beta}^{2}(\xx)\, d\xx &= \beta^{2} \int_\mathbb{R} W_0(\beta^{1/\mu}\xx)\,d\xx
  = \beta^{2-1/\mu}\int_\mathbb{R} W_0^2(z)\, dz\,.
\end{align}
Hence at $\beta =1$,
\begin{equation}
  \frac{d}{d \beta} \int_\mathbb{R} W_{0,\beta}^2\, d\xx = \int_\mathbb{R} 2 W_0 \frac{\partial}{\partial \beta} W_{0,\beta}\, d\xx = \frac{2\mu -1}{\mu} \int_\mathbb{R} W_0^2 \,d\xx\,.
\end{equation}
2. From differentiating the traveling wave equations, \eqref{e:beta_evp3}
for $W_{\eps,\beta}$ and \eqref{e:beta_W0} for $W_{0,\beta}$, against $\beta$, we get 
\begin{align}
      V_\eps:=           \frac{\partial}{\partial \beta} W_{\eps,\beta} \Big|_{\beta = 1} &= - (I - D\calG_\eps(W_\eps))^{-1}(\calB_\eps^{-1}W_\eps)\,,\\
 V_0:=\frac{\partial }{\partial \beta} W_{0,\beta} \Big|_{\beta =1} &= - (I - D\calF(W_0))^{-1}(\calB_0^{-1}W_0)\, .
\end{align}
The convergence $V_\eps \to V_0$ in $H^{1}$ is obtained through the operator norm convergence 
$\calB_\eps^{-1}\to \calB_0^{-1}$ and the estimates in Proposition~\ref{p:DFG}. 
One should note that we lack a rate of convergence due to our result for  
$\calB_{\eps}^{-1} \to \calB_0^{-1}$.
\end{proof}

\begin{lemma}
We have that as $\eps \to 0$,
\begin{equation}
  \frac{d}{dc}\calH(u_c) = 2c_\alpha^{3}\eps^{\mu-1}\frac{2\mu - 1}{\mu} \int_\mathbb{R} W_0^{2}\, d\xx + o(\eps^{\mu-1})\,.
\label{e:dhdc}
\end{equation}
\end{lemma}
\begin{proof}
Using the definition \eqref{d:Wepsbeta} of $W_{\eps,\beta}$
and with the scaling $c^{2}= c_\alpha^{2}+\beta \eps^{\mu}$, 
from \eqref{e:calIuc} we get
\begin{align}
   \calI(u_c) 
    &= -c \int_\mathbb{R}  (\eps^{\mu}W_{\eps,\beta}(\eps z))^2\, dz
    = -c  \eps^{2\mu-1}\int_{\R}W_{\eps,\beta}^2(\xx)\, d\xx \, .
\end{align}
Then, fixing $\eps$ and differentiating in $\beta$ at $\beta=1$, 
since $d\beta/dc = 2c\eps^{-\mu}$ we find 
\begin{align}
    \frac{d}{dc} \calI(u_c)
     &= -\eps^{2\mu-1}\int_\mathbb{R} W_{\eps,\beta}^2(z)\, d\xx 
     - c \eps^{2\mu-1}
     \int_\mathbb{R} 2W_{\eps,\beta} \frac{\partial}{\partial \beta} W_{\eps,\beta} 
     \, d\xx\, \frac{d\beta}{dc}  \Big|_{\beta =1} 
     \nonumber \\
    &= O(\eps^{2\mu-1}) -2c^2 \eps^{\mu-1}\int_\mathbb{R} 2W_\eps\, V_\eps \, d\xx 
  \,.
\end{align}
Recalling that $\delta \calH = -c \delta \calI$, we have
\begin{equation}
               \frac{d}{dc} \calH(u_c) = 2c^{3}\eps^{\mu-1}\int_\mathbb{R} 
	       2 W_\eps\, V_\eps \,d\xx 
	       + O(\eps^{2\mu -1})\,. 
\end{equation}
Expanding $c^{2}=c_\alpha^2+\eps^\mu$ and using the previous lemma gives 
\begin{align}
    \frac{d}{dc} \calH (u_c) &= 
    2c_\alpha^{3}\eps^{\mu-1} \int_\mathbb{R} 
    2 W_0\,V_0\,d\xx
    + o(\eps^{\mu-1})
     \nonumber \\
   &= 2c_\alpha^{3}\eps^{\mu-1}\frac{2\mu-1}{\mu} \int_\mathbb{R} W_0^{2}\, d\xx + o(\eps^{\mu-1})\,.
\end{align}
This completes the proof. 
\end{proof}

Now, through multiplying \eqref{e:dhdc} by 
\[
\frac{dc}{d\eps} = 
\frac{\mu\eps^{\mu-1}}{2c} =  
\frac{\mu\eps^{\mu-1}}{2c_\alpha}   
+O(\eps^{2\mu-1}) \,,
\]
we deduce \eqref{e:dHdeps}.
This completes the proof of Theorem~\ref{t.ham}.
\end{proof}

We conclude by calculating the Hamilonian in the case  of the Calogero-Moser lattice when $\alpha=2$.
\begin{proof}[Proof of Theorem~\ref{t.CMH}]
We have $\calH(u_c)\to0$ as $c\to\pi^+$ by Proposition~\ref{p:Ham}, so the claimed formula
$\calH(u_c)=\frac12(c^2-\pi^2)$ follows by integration from the formula
\begin{equation}
    \frac{d\calH}{dc} = -c\frac{d\calI}{dc} = c\,,
\end{equation}
which holds due to the following computation.
\end{proof}
\begin{lemma}
When $\alpha=2$, for all $c>\pi$  we have $\calI(u_c) = \pi -c$.
\end{lemma}
\begin{proof} 
For $\alpha=2$, the solitary waves satisfy $x_j(t)=j-\vp(j-ct)$
with $\vp(z)$ satisfying \eqref{e:explicit_eq}. Hence $q(z)=\vp(z)$,
so from \eqref{d:calI} and \eqref{e:qp1} it follows
\[
\calI(u_c) = -c\int_\R \left(\frac{d\vp}{dz}\right)^2\,dz = 
-c\int_{-1/2}^{1/2} \frac{d\vp}{dz}\,d\vp \,,
\]
since $\vp\to \pm\frac12$ as $z\to\pm\infty$. Differentiating \eqref{e:explicit_eq}, we see
\[
(c^2-\pi^2) = \frac{d\vp}{dz}(c^2 +\pi^2 \tan^2\pi\vp) \,,
\]
whence
\[
\calI(u_c) = 
-c\int_{-1/2}^{1/2}  \left(1 - \frac{\pi^2\sec^2\pi \vp}{c^2+\pi^2\tan^2\pi\vp} \right)\,d\vp \,.
\]
Using the substitution $c y = \pi\tan\pi \vp$ one finds the claimed result.
\end{proof}

\section*{Acknowledgements}
  This material is based upon work supported by the National Science Foundation under
  Grant No.~DMS 2106534. Thanks go to Doug Wright for very helpful discussions.

\bibliographystyle{siam}
\bibliography{CM2}

\end{document}